\documentclass[twoside,leqno,twocolumn]{article}
\pdfoutput=1

\usepackage{ltexpprt}

\usepackage{authblk}
\usepackage{lmodern}
\usepackage{balance}

\usepackage{amsmath,amssymb}

\usepackage[margin=1in]{geometry}

\usepackage[colorlinks,
linkcolor=blue,
filecolor=blue,
citecolor=magenta,
urlcolor=blue,
pdfstartview=FitH,
breaklinks]{hyperref}

\usepackage{url}
\usepackage{bm}
\usepackage{multirow}
\usepackage{graphicx}
\usepackage[boxed]{algorithm}

\usepackage{mathptmx}
\usepackage{latexsym}
\usepackage{verbatim}
\usepackage[dvipsnames,usenames]{color}
\usepackage{float}
\usepackage{todonotes}
\usepackage{enumerate}
\usepackage{pgf,tikz}
\usetikzlibrary{arrows}

\usepackage{epigraph}

\usepackage{etoolbox}

\makeatletter
\patchcmd{\epigraph}{\@epitext{#1}}{\itshape\@epitext{#1}}{}{}
\makeatother

\newtheorem{remark}{Remark}[section]
\newtheorem{observation}{Observation}[section]

\newcommand{\lp}{\left(}
\newcommand{\rp}{\right)}

\newcommand{\la}{\left|}
\newcommand{\ra}{\right|}

\DeclareMathOperator*{\argmin}{arg\,min}
\DeclareMathOperator*{\argmax}{arg\,max}

\newcommand{\eat}[1]{}

\newcommand{\A}{A}

\newcommand{\jmin}{j_{\mathrm{min}}}

\DeclareMathOperator{\E}{E}

\newcommand{\Reals}{\mathbb{R}}

\newcommand{\email}[1]{\href{mailto:#1}{\nolinkurl{#1}}}
\newcommand{\qed}{\hfill\hbox{\rlap{$\sqcap$}$\sqcup$}}
\newcommand{\qedd}{\hfill\hbox{\rlap{$\sqcap$}$\sqcup$\ \rlap{$\sqcap$}$\sqcup$}}

\title{Pricing Online Decisions: Beyond Auctions}

\author[1]{Ilan Reuven Cohen}
\author[1]{Alon Eden}
\author[1]{Amos Fiat}
\author[1,2]{\L{}ukasz Je\.z{}}
\affil[1]{Blavatnik School of Computer Science, Tel-Aviv University%
\thanks{Partially supported by the Israeli Centers of Research Excellence (I-CORE) program, (Center  No.4/11) and the Google Inter-University Center.}
}
\affil[2]{University of Wroc{\l}aw, Institute of Computer Science%
\thanks{Partially supported by the Foundation for Polish Science (FNP) START Scholarship
and the Polish National Science Center (NCN) Grant DEC-2013/09/B/ST6/01538.}
}
\affil[ ]{\textit {\{alonarden,ilanrcohen\}@gmail.com, fiat@tau.ac.il, lje@cs.uni.wroc.pl} }

\floatstyle{ruled}
\newfloat{alg}{thbp}{loa}
\floatname{alg}{Algorithm}

\newfloat{pricing}{thbp}{lop}
\floatname{pricing}{Pricing Scheme}

\newfloat{linprog}{thbp}{llp}
\floatname{linprog}{Linear Program}

\newfloat{example}{thbp}{loe}
\floatname{example}{Example}

\begin{document}
\definecolor{dcrutc}{rgb}{0.86,0.08,0.24}
\definecolor{qqqqff}{rgb}{0,0,1}
\definecolor{cqcqcq}{rgb}{0.75,0.75,0.75}
\definecolor{ffqqtt}{rgb}{1,0,0.2}
\definecolor{ffttqq}{rgb}{1,0.2,0}
\definecolor{xdxdff}{rgb}{0.3,0.4,0.5}
\maketitle

\epigraph{``You pays your money and you takes your choice."}{--- \textup{Mark Twain}, Huckleberry Finn}


\begin{abstract}
We consider dynamic pricing schemes in online settings where selfish agents generate online events. Previous work on online mechanisms has dealt almost entirely with the goal of maximizing social welfare or revenue in an auction settings. This paper deals with quite general settings and minimizing social costs. We show that appropriately computed posted prices allow one to achieve essentially the same performance as the best online algorithm. This holds in a wide variety of settings. Unlike online algorithms that learn about the event, and then make enforceable decisions, prices are posted without knowing the future events or even the current event, and are thus inherently dominant strategy incentive compatible.

In particular we show that one can give efficient posted price mechanisms for metrical task systems, some instances of the $k$-server problem, and metrical matching problems. We give both deterministic and randomized algorithms. Such posted price mechanisms decrease the social cost dramatically over selfish behavior where no decision incurs a charge. One alluring application of this is reducing the social cost of free parking exponentially.
\end{abstract}

\section{Introduction}
We consider scenarios where agents arrive and make decisions which influence the global environment. One such setting is that of online auctions, where goods or combinations of goods are sold over time. In this paper we consider a more general setting that captures more general online social choice issues.

We suggest using dynamic posted prices that dynamically assign a surcharge to every decision. Under the assumption that agents are rational, this allows one to prove approximation bounds with respect to optimal target function.

The competitive analysis of online algorithms typically deals with problems cast as follows:
\begin{itemize}
  \item A sequence of events $\sigma_1, \sigma_2, \ldots$ occurs over time.
  \item Subsequent to observing event $\sigma_i$, a centralized online algorithm takes some decision $s_i\in S_i$, where $S_i$ is the set of possible decisions that can be taken in response to event $\sigma_i$.
  \item Centralized online algorithms choose their decisions for event $i$ without knowing what the future holds so as to maximize (or minimize), some function $f(\overline{\sigma},\overline{s})$ of the sequence of events $\overline{\sigma}=\sigma_1, \ldots, \sigma_n$ and associated decisions $\overline{s}=s_1, \ldots, s_n$.
\end{itemize}

In this paper we reinterpret the framework above by having strategic agents initiate events.
So rather than an online sequence of events we have an online sequence of strategic agents,and rather than a centralized online algorithms making decisions the strategic agents make the decisions themselves.

The $i$th agent, has some (private) type $t_i \in T_i$.
Conditioned upon previous events $\sigma_1, \ldots, \sigma_{i-1}$, and previous decisions $s_1\in S_1, \ldots, s_{i-1}\in S_{i-1}$, an agent type $t\in T_i$ is a function  $t:S_i\mapsto \Reals.$ An agent of type $t\in T_i$ has value $t(s)$ for decisions $s\in S_i$.  Ergo, a rational agent $i$ will choose decision $$s_i \in \argmax_{s\in S_i} t_i(s).$$ We remark that $t$ may take either positive or negative values. In the translation of online problems to the strategic setting, costs translate to negative value, benefits translate to positive values. Note that the target function $f$ remains as defined in the non-strategic online setting.

To motivate these strategic agents to ``do the right thing" (from the perspective of the designer), we introduce {\sl dynamic posted pricing schemes} that set a surcharge on the set of possible decisions. A dynamic posted pricing scheme $P$ generates a sequence of functions $\langle P_1, P_2,
\ldots \rangle$. Dynamic posted pricing schemes determine the $i$'th function, $P_i$,  {\em before} event $\sigma_i$ arrives.

The function $P_i:S_i \mapsto \Reals$ sets posted prices (a surcharge) for every possible decision in $S_i$. We assume quasilinearity, {\em i.e.},  the utility of agent $i$ of type $t\in T_i$ --- should she decide upon $s\in S_i$ --- is $$u_i(s) = t(s)-P_i(s).$$ Rational agents always seek to maximize their utility, or minimize their disutility which is the negation of their utility.

Dynamic posted prices of particular interest are varying tolls on express lanes and the SFPark system (SFpark.org) which sets parking prices in San Francisco as a function of parking congestion \cite{shoup2005high,NYT:sfpark}. Interestingly, our dynamic posted pricing scheme for metric matching on the line described below can be viewed as a refinement of the SFPark system.

A very simple example for dynamic pricing, where prices are set once and for all (all $P_i$'s are the same), is the problem of online matching on unweighted bipartite graphs\cite{KarpVV90}. In this problem, there is a stationary right side, vertices (agents) on the left arrive online along with edges to all adjacent vertices on the right. The target is to maximize the number of matches made.

The KVV algorithm \cite{KarpVV90} performs a random permutation of the vertices on the right and assigns an incoming vertex to the first available vertex on the right. Translating this into the strategic agent setting, agent types give a utility of one for a match with one of the adjacent vertices. Strategic agents can be induced to emulate the KVV algorithm by a pricing scheme assigning a surcharge of $\epsilon, 2 \epsilon, \ldots$ to the vertices on the right in some random order. Note that the surcharge for a vertex never changes so there is no ``dynamic" aspect to this pricing scheme.

Posted prices are inherently a dominant strategy incentive compatible mechanism. Given a dynamic posted pricing scheme $P$ that produces functions $\langle P_1, P_2,
\ldots \rangle$, and a sequence of agent types $\overline{t}=t_1, t_2, \ldots, t_n$, the set of decision sequences in equilibria is $$\mathrm{eq}(P,\overline{t}) = \left\{ \overline{s} \,\middle\vert\, \overline{s} = s_1, \ldots, s_n,\quad s_i\in \argmax_{s\in S_i} u_i(s)\right\}.$$
Because of ties it may be that  $|\mathrm{eq}(P,\overline{t})|>1$.

In the context of mechanism design it is often assumed that if a mechanism is dominant strategy truthful, the agents may as well reveal their true type to the mechanism. Moreover, the revelation principle states that without loss of generality, one need consider only such direct revelation mechanisms.

In many contexts, revealing the agent type is problematic or absurd. When going to the supermarket one may observe what people are buying, it is hardly plausible to interview customers at length as to what their true preferences are. Moreover, revealing the true agent type may be a grave breach of privacy, see \cite{Naor-priv}.

One could, in principle, combine a dynamic posted pricing scheme yet insist that the agents reveal their true type. Were we to allow this, the functions $P_i$ could depend on the true types of previous agents: $t_1, \ldots, t_{i-1}$, and on previous decisions taken by these agents: $s_1, \ldots, s_{i-1}$ (may be relevant if agents use randomization).

 To avoid the unnatural aspects of type revelation, we like to restrict the information available to the dynamic pricing scheme. This should be information that is (relatively) easy to collect by simple observation of the consequence of agent decisions without having to understand the causes for such decisions.  {\em E.g.}, one can observe what customers purchase, but not what they plan to cook for dinner or whom they are having over. We would like to allow dynamic pricing schemes\footnote{Or mechanisms in general.} to observe (some) consequences of decisions made by (earlier) agents, but not necessarily allow access to the true agent types.

The less one assumes about observable history, the more robust the concept. {\em E.g.}, in the context of dynamic posted pricing for parking, it seems reasonable to place sensors at parking spaces so as to as to learn what parking spaces are free but it seems much less plausible to require a host of miniaturized spy drones to follow everyone about (so  as to learn the true agent type --- her true destination). See discussion and definitions in Section \ref{sec:discuss}.

In the context of strategic agents, we use the term price of anarchy \cite{KP-POA} instead of the term competitive ratio. {\em I.e.}, under the assumption that agents are rational, the price of anarchy is the ratio between the value of the target function $f$ when making the best possible decisions and the value of $f$ when decisions are the worst case equilibria.

We distinguish between benefit problems, where the goal is to maximize $f$, and cost problems where the goal is to minimize $f$. In particular, social welfare maximization is to maximize the function $\sum_{i=1}^{|\sigma|} t_i(s_i)$, and social cost minimization is to minimize the  social cost function $ -\sum_{i=1}^{|\sigma|} t_i(s_i).$ Clearly the goals are equivalent.

Given a dynamic pricing scheme $\overline{P}$, target function $f$, the price of anarchy (for benefit problems) is  \begin{equation}\min_{\overline{\sigma},\, \overline{s}\in \mathrm{eq}(\overline{P},\overline{\sigma})}\frac{f(\overline{\sigma},\overline{s})}{  f(\overline{\sigma},\overline{\mathrm{opt}(\sigma)})}, \label{eq:poab}\end{equation} where $$\overline{\mathrm{opt}(\sigma)}\in \argmax_{\overline{s}: s_i\in S_i \text{ for all } i} f(\overline{\sigma},\overline{s}).$$

 Given a dynamic pricing scheme $\overline{P}$, and target function $f$, the price of anarchy (for cost problems) is  \begin{equation}\max_{\overline{\sigma},\, \overline{s}\in \mathrm{eq}(\overline{P},\overline{\sigma})}\frac{f(\overline{\sigma},\overline{s})}{  f(\overline{\sigma},\overline{\mathrm{opt}(\sigma)})}, \label{eq:poa}\end{equation} where $$\overline{\mathrm{opt}(\sigma)}\in \argmin_{\overline{s}: s_i\in S_i \text{ for all } i} f(\overline{\sigma},\overline{s}).$$

It is common to consider the competitive ratio in terms of an adversary, and one can do the same in the context of the price of anarchy. With this interpretation, one may assume that the adversary knows the true types of the entire sequence in advance.

To clarify the differences between online algorithms, online mechanisms, and online posted prices (a special case of online mechanisms) the following summary may be helpful:

\begin{enumerate}
\item Online algorithms: events appear and are dealt with in a centralized fashion. It is assumed that events are ``true". Online algorithms enforce their decisions. There is no sense in which an event says ``I don't like what you've decided in my case."
\item Online mechanisms: events appear and decisions are made, prizes and fees set, in a centralized fashion.  Just as with competitive algorithms, online mechanisms can (magically?) enforce their decisions. There is no sense in which an event can ``revolt" against a decision that she does not like.
\item  Dynamic posted pricing (this paper) sets a posted price for every possible decision. Incoming strategic events make their own decisions based upon their own self interest. ``You pays your money and you makes your choice". Moreover, different dynamic posted pricing schemes can be compared both with respect to the competitive ratio they achieve and with respect to the quality of the observable history they require.
\end{enumerate}

\subsection{Problems and Results}

We consider the following problems

\subsubsection*{Metrical task systems \cite{BorodinLS92}}

    Metrical task systems are described by a symmetric non-negative $m\times m$ matrix, $(d_{st})_{1 \leq s,t \leq m}$, where $d_{ss}=0$ and the triangle inequality holds: $d_{st}\leq d_{sk} + d_{kt}$ for all $1 \leq s,t,k \leq m$. The set $S=\{1, \ldots, m\}$ is the set of states. A task is a vector of $m$ non negative real values, $${w} =\left(w_1, w_2, \ldots, w_m\right).$$
      In response to the arrival of a new task, $w^i$,
      a centralized online algorithm decides upon some state $s\in S$.
      The cost of a sequence of decisions $s_1, \ldots, s_n$ for a sequence of tasks $w^1, \ldots, w^n$ is $$\sum d_{s_{i-1}s_i} + w^i_{s_i}.$$

      \subsubsection*{Strategic metrical task systems}
          The [dis]utility of an agent is the sum of transition cost (to go from state to state), the cost of executing the task in the state chosen, and the surcharge for the decision taken. It is very easy to see that, without introducing a surcharge, the price of anarchy is unbounded. strategic agents will never take expensive transitions. If one were to consider an online algorithm as a mechanism without money, agents will lie shamelessly. Every state other than their greedy preference would cost infinity.

Here we give a dynamic posted pricing scheme with a competitive ratio of $O(m)$, within $O(1)$ of the best deterministic competitive ratio possible. The observable history consists of the actual state transitions and the cost to serve the tasks in the states chosen. The observable history {\sl does not} include the cost of the tasks in states that were not chosen.

   \subsubsection*{The $k$-server problem \cite{ManasseMS90,KP-server,Koutsoupias09}}

   The $k$ server problem has $k$ servers located in some metric space. Online requests arrive following which
some server is to be moved to the request location. {\em I.e.}, one needs to decide upon which server is to be moved. The goal is to minimize the sum of distances traveled.
A competitive ratio of $k$ is known for some metric spaces, whereas $2k-1$ is an upper bound for general metric spaces \cite{KP-server}, {\em vs.} a lower bound of $k$. Closing the gap for deterministic algorithms is an infamous open problem.

  \subsubsection*{The strategic $k$-server problem}

The disutility of an agent is the distance traversed by the server. Thus, without a dynamic posted pricing scheme, agents will always choose the closest server to move, giving an unbounded price of anarchy. We
give a dynamic posted pricing scheme for any metric space with 2 servers with  a price of anarchy 10. For $k$-servers on the line, we give the optimal price of anarchy $k$. In this case, the observable history (the server positions) allows one to infer the true type of previous agents.

\subsubsection*{Metric Matching~\cite{KhullerMV94,DBLP:conf/dagstuhl/KalyanasundaramP96}}

       The metric matching problem is described on some metric space. Events are points in the metric space, who are to be matched to some unmatched point, with the goal of minimizing the sum of distances between the incoming points and their match. One application of this problem is parking, where one seeks to assign homeowners a parking spot close to their home. The problem on the line is another infamous open problem, conjectured to have a constant competitive ratio in \cite{DBLP:conf/dagstuhl/KalyanasundaramP96} but still open.

\subsubsection*{Strategic Metric Matching}

 The disutility of an agent is the distance between the agent destination and the point assigned to the agent. Ergo, strategic agents will always choose the closest possible point.

   It follows from \cite{KhullerMV94} that without surcharges the price of anarchy may be exponential in the size of the metric space.

We give a randomized dynamic posted pricing scheme with a competitive ratio logarithmic in the ratio between the largest and smallest distances between points. In particular, this implies an exponential improvement in the price of anarchy over the price of anarchy without surcharges. Moreover, given an estimate on the cost of the optimal solution, we can attain a doubly exponential improvement over the price of anarchy without surcharges.


\section{Related Work}

Lavi and Nisan \cite{LN00} initiated the study of online auctions. The consider a model of indivisible goods where bidders have a privately known valuation for different quantities of the goods. Agents learn their valuation at some time and must bid immediately. The auction must decide immediately how many units to allocate and at what price. They show that the only dominant strategy incentive compatible mechanisms are those that offer a price menu for different quantities.

Awerbuch, Azar, and Myerson~\cite{AAM03} give a general scheme that produces posted prices for general combinatorial auctions, with a competitive ratio equal to the logarithm of the ratio between highest and lowest prices, times the underlying competitive ratio for the combinatorial auction (which is terrible, in general).

In \cite{DBLP:conf/sigecom/FriedmanP03} Friedman and Parkes introduced a model of online mechanism design, they seek strategyproof mechanisms where agents may manipulate their arrival time and values. They give a variant of the VCG mechanism for some such problems.
One of the issues from this paper and subsequent work has been misrepresentation of arrival times (or departure times, or both).

In the main, we do not consider misrepresentation with respect to time or position in the sequence. See discussion in Section \ref{sec:discuss}

In
Parkes~\cite{Parkes} a general formalism for dynamic environments and online mechanism design is given. In this general model a mechanism makes (and enforces) a sequence of decisions. The decisions may depend on interaction with agents as well as changes to the environment.
In the online mechanism design model, agent types include arrival times, departure times, and a valuation function of the decisions made by the algorithm.

 Unfortunately, the positive results in the general setting of \cite{Parkes} have been limited to so-called ``single valued preferences" where an agent gets a fixed value $r$ iff  one of a set of ``interesting decisions" in made between her arrival and departure. (The positive results also require limited misrepresentation of the agent types).

It is important to note that none of the problems we consider are single valued. For  task systems, the $k$-server problem, and the parking problem, different decisions will generally give different value to the agent. Also note that we do not allow misrepresentation of arrival time, and insist that departure time be equal to arrival time.

There has been a large body of work on online auctions in the setting where the order of events is a uniformly random permutation. Hajiaghayi, Kleinberg and Parkes \cite{HKP04} discuss limited supply online auctions, and Babaioff, Immorlica, Kempe, and Kleinberg \cite{BabaioffIKK08} introduced the knapsack and matroid secretary problems, about which there has been much subsequent work.

\section{Pricing States for Metrical Task Systems}
\subsection{Problem Definition}
A metrical task system \cite{BorodinLS92} is described by a symmetric non-negative $m\times m$ matrix, $(d_{st})_{1 \leq s,t \leq m}$, where $d_{ss}=0$ and the triangle inequality holds: $d_{st}\leq d_{sk} + d_{kt}$ for all $1 \leq s,t,k \leq m$. The set $S=\{1, \ldots, m\}$ is the set of states.

The strategic metrical task system is a variant of the metrical task system where tasks are associated with agents.
Agents arrive online, the $i$th agent to arrive has some task to perform, represented by a vector of $m$ non negative real values, $${w}^i =\left(w^i_1, w^i_2, \ldots, w^{i}_m\right).$$ The task $w^{i}$ is private to agent $i$.

The initial state of the system is $s_0$, the state of the system may change over time. Let $s_{i-1}$ be the state of the system upon the arrival of agent $i$. The $i$th agent to arrive may decide to change the state of the system upon her arrival to some other state ({\em i.e.}, $s_i \neq s_{i-1}$) or not ({\em i.e.} $s_{i}=s_{i-1}$).

Given states $s$, $t$, and a task $w=(w_1,w_2, \ldots, w_m)$, define $$C(s,t,w) = w_{t} + d_{st},$$ {\em i.e.} the cost to switch states from $s$ to $t$ and process task $w$ in state $t$.

A dynamic pricing scheme for strategic metrical task systems sets prices (surcharges) on the states of the system. Formally, a dynamic pricing scheme $P=\left(P_1,P_2,\ldots \right)$ is a sequence of price functions $P_1, P_2, \ldots$, where the $i$th function  $$P_{i}:\left\langle\left(s_1,w^1_{s_1}\right),\left(s_2,w^2_{s_2}\right), \ldots, \left(s_{i-1},w^{i-1}_{s_{i-1}}\right)\right\rangle\times S\mapsto \Reals^+.$$ The information available to the pricing function $P_{i}$ is the state $s_j$ chosen made by agent $j$ and the work $w^j_{s_j}$ expended on the $j$th agent's task $w^j$ in that state, for all $j < i$. This is captured by the sequence $\left\langle\left(s_1,w^1_{s_1}\right),\left(s_2,w^2_{s_2}\right), \ldots, \left(s_{i-1},w^{i-1}_{s_{i-1}}\right)\right\rangle$. The prices computed by pricing function $P_{i}$ are relevant for agent ${i}$.

For notational convenience, we use $P_{i}$ as a pricing function over the set of states after $i-1$ tasks are processed, where the ``history" of the system is implicit. We stress that the pricing is set without any information regarding future events.

Note that pricing function $P_i$ {\em cannot} depend on tasks $w^j$, $j\geq i$, (this is distinct from an online algorithm  that may depend on $w^{i}$). Moreover, it knows very little about $w^1, w^2, \ldots, w^{i-1}$. It only knows the {\em decisions} made by the agents and the actual work done, {\em i.e.}, $s_j$ and $w^j_{s_j}$ for $1 \leq j \leq i-1$. In particular, this means that the pricing scheme does not know the optimal cost to process tasks $w^1,\ldots,w^{i-1}$ and cannot compute the work function for a state.

The goal of a strategic agent is to to minimize her disutility. This is the sum of the following components:
\begin{enumerate}
  \item The work required to do the task in the state selected by the agent. For agent $i$ this is $w^i_{s_i}$.
  \item The transition cost to change states. For agent $i$ this is $d_{s_{i-1}s_i}$.
  \item The surcharge associated with choosing (or remaining in) some state, this is determined by the pricing scheme (see below). For agent $i$ this is $P_{i}(s_i)$.
\end{enumerate}

Therefore, the agent will choose:
$$s_{i} \in \argmin_{s \in S_i} \{ w^i_{s} + d_{s_{i-1}s} + P_{i}(s) \}$$
As the goal in metrical task systems is cost minimization, the target function $f$ is the social cost, $$f(\overline{\sigma},\overline{s}) = \sum_{i=1}^{|\overline{\sigma}|} \left( w^i_{s_i} + d_{s_{i-1}s_i}\right),$$  {\em i.e.}, the sum of all state transitions and task processing costs.  The price of anarchy is as defined in Equation (\ref{eq:poa}).

\subsection{The Fractional Traversal Algorithm for Metrical Task Systems}\label{sec:travdesc}

We now describe the fractional traversal algorithm for metrical task systems from \cite{BorodinLS92}. The fractional traversal algorithm solves a relaxation of the original problem in that it allows states to perform fractions of a task, possibly switching states repeatedly until the task is completed.

A traversal sequence, $\tau=\tau_1,\tau_2,\ldots$, is an infinite sequence of states of the task system, where every state of the task system may appear several times in the sequence. The fractional traversal algorithm given here as Algorithm~\ref{alg:t} is very simple, it executes a fraction of task until the work done in the current state, since moving to this state, equals the distance to move to the next state in the sequence.

Borodin et al. showed the following with regard to the traversal algorithm:
\begin{theorem}[Borodin et al.~\cite{BorodinLS92}]
\begin{itemize}
   \item There exists a traversal sequence such that the fractional traversal algorithm has a competitive ratio of $8(m-1)$. (Versus the optimal competitive ratio of $2m-1$ achieved by the work function algorithm).
   \item Given any online algorithm that schedules tasks fractionally on multiple states with competitive ratio $c$, there exists another online algorithm that executes the task in a single state with competitive ratio $\leq c$.
  \end{itemize}
\end{theorem}

 The following variables are used in the traversal algorithm (Algorithm~\ref{alg:t}):

 \begin{itemize}
   \item The current position in the traversal sequence $\tau$ is indicated by the variable $j$.
   \item Task $w^i$ is (fractionally) executed in states $$\tau_{t_{i-1}}, \tau_{t_{i-1}+1}, \ldots, \tau_{t_i}.$$ {\em I.e.}, during the processing of task $w^i$ the index $j$ takes values $t_{i-1}, \ldots, t_i$.
   \item The fraction of task $w^i$ that is executed in state $\tau_j$ is denoted by
   $\lambda^i_{j}$. Note that $\lambda^i_\ell=0$ for all $\ell \notin t_{i-1}, \ldots, t_i$,
            Note too that there may be $\ell,\ell'\in t_{i-1}, \ldots, t_i$, $|\ell-\ell'|\geq 2$, such that $\tau_\ell=\tau_{\ell'}$.
   \item The variable $\rho_j$ when dealing with task $i$ represents the work expended thus far in the $j$th position of $\tau$. {\sl I.e.},
   $$\rho_j = \sum_{i': t_{i'} \leq j \leq t_{i'+1}, i'\leq i} \lambda^{i'}_j w^{i'}_{\tau_j}.$$
 \end{itemize}

\begin{alg}[tb]
    \textbf{Input:} A traversal sequence $\tau$, a set of states $S$, the function ${d}$, and an online sequence of tasks $w_1,w_2,\ldots $.
    For each task $i$:

    $j \gets t_{i-1}$

    while ($\sum_{k=t_{i-1}}^j \lambda^i_{k} < 1$)
    \begin{itemize}
        \item $\lambda^i_{j} \gets \min\left\{\frac{d_{j,j+1}-\rho_{j}}{w^i_j},1-\sum_{k=t_{i-1}}^{j-1} \lambda^i_{k}\right\}$
        \item $\rho_{j} \gets \rho_{j} + \lambda^i_{j}\cdot w^i_{\tau_j}$
        \item if $ \rho_{j} = {d}_{\tau_j,\tau_{j+1}}$ then $j \gets j+1$, $\rho_{j}\gets 0$
    \end{itemize}
    $t_i \gets j$

    \caption{The Fractional Traversal Algorithm.} \label{alg:t}
\end{alg}

An illustration of the execution of Algorithm \ref{alg:t} is given in Appendix \ref{append:trav-ex}.


Given a traversal sequence $\tau=\tau_1,\tau_2,\ldots$, define the traversal distance between indices $\ell$ and $\ell'$ to be $$\delta_{\ell,\ell'}=\delta_{\ell',\ell}=\sum_{j=\min({\ell,\ell'})}^{\max({\ell,\ell'})-1}{d}_{\tau_j,\tau_{j+1}}.$$

    Note  that \begin{itemize}
\item The total work done by the fractional traversal algorithm after completing $n$ tasks is exactly $\delta_{t_0,t_n}+\rho_{t_n}$. This follows as the algorithm advances to the next state in the traversal sequence when the amount of work done in the current state is equal to the distance to the next state.
\item The total traversal cost to the fractional traversal algorithm for $n$ tasks is exactly $\delta_{t_0,t_n}$, every transition from $\tau_j$ to $\tau_{j+1}$ has transition cost $d_{j,j+1}.$
    \end{itemize}

    We now show that a Traversal Algorithm's performance is monotone with respects to the tasks it receives.
    \begin{lemma} \label{lem:travmon}
        Let $\tau$ be some traversal sequence, and let $\mathbf{w}=\langle w^1,\ldots,w^n\rangle$ and $\mathbf{\tilde{w}}=\langle\tilde{w}^1,\ldots \tilde{w}^n\rangle$ be two sequences of tasks such that for every $i\in \{1,\ldots, n\}$ and every $s \in S$, we have that $\tilde{w}^i_s\leq w^i_s$. The total work done by the traversal algorithm with respect to $\tau$ and $\mathbf{w}$ is at least the amount of work done with respect to $\tau$ and $\mathbf{\tilde{w}}$.
    \end{lemma}
    \begin{proof}
        Let $t_i$ be the position of the traversal algorithm after processing task $i$ in the task sequence $\mathbf{w}$, and $\tilde{t}_i$ the position of the traversal algorithm after processing task $i$ in the task sequence $\mathbf{\tilde{w}}$. In addition, let $\rho_{t_i}$ be the amount of work done by the algorithm since index $t_i$ given sequence $\mathbf{w}$, and $\tilde{\rho}_{\tilde{t}_i}$ the corresponding amount of work since reaching $\tilde{t}_i$ and given task sequence $\mathbf{\tilde{w}}$.

We show by an inductive argument that for every $i\in \{0,\ldots, n\}$, $\delta_{t_0,t_i}+\rho_{t_i}\geq \delta_{\tilde{t}_0,\tilde{t}_i}+\tilde{\rho}_{\tilde{t}_i}$. This is clearly true for $i=0$. Let the claim hold before the arrival of the $i$-th task. If the traversal algorithm that processes sequence $\mathbf{\tilde{w}}$ doesn't reach $t_i$ when processing task $\tilde{w}^i$, ($\tilde{t}_i < t_i$) then
\begin{eqnarray*}
  \delta_{t_0,\tilde{t}_i}+\rho_{\tilde{t}_i} &\leq& \delta_{t_0,\tilde{t}_i +1} \\
  &\leq& \delta_{t_0,t_i} \\
  &\leq& \delta_{t_0,t_i}+\rho_{t_i}.
\end{eqnarray*}
where the 1st derivation above is by the definition of the algorithm, the 2nd is by assumption that $\tilde{t}_i < t_i$ and the 3rd is since $\rho_j\geq 0$ for all $j$.

Lets assume that $\tilde{t}_i \geq t_i$.  We'll show that it cannot be that $\tilde{t}_i > t_i$, and when $\tilde{t}_i=t_i$ we show
that $\tilde{\rho_{t_i}}\leq \rho_{t_i},$ thus concluding the inductive proof.

 Because $w^i_s\geq \tilde{w}^i_s$ for every state $s$, we have that for every $j\in \{t_{i-1},\ldots, t_{i}-1\}$, \begin{equation}\tilde{\lambda}^i_j\geq  \lambda^i_j,\label{eq:greaterlambda} \end{equation} where $\tilde{\lambda}^i_j$ and $\lambda^i_j$ are the fractions of tasks $\tilde{w}_i$ and $w^i$ done in $\tau_j$. This follows from the definition of $\lambda$ in the fractional traversal algorithm and the inductive hypothesis.

Therefore,
\begin{eqnarray} \tilde{\lambda}^i_{t_i} &\leq& 1 - \sum_{j=t_{i-1}}^{t_i-1} \tilde{\lambda}^i_j \label{eq:leftover} \\
 &\leq& 1 - \sum_{j=t_{i-1}}^{t_i-1} {\lambda}^i_j \label{eq:lambdasmall} \\
 &=& \lambda^i_{t_i}. \nonumber\end{eqnarray}
 Inequality (\ref{eq:leftover}) follows since the fraction left over cannot be more that one minus the fractions already done.
 Inequality (\ref{eq:lambdasmall}) follows from Equation (\ref{eq:greaterlambda}), and the final equality is by definition.

  Therefore, $$\tilde{\rho}_{t_i} = \tilde{\lambda}^i_{t_i}\cdot \tilde{w}_i\left(\tau_{t_i}\right) \leq \lambda^i_{t_i}\cdot w^i_{\tau_{t_i}} = \rho_{t_i} < {d}_{t_i,t_i+1},$$ so $\tilde{t}_i = t_i$, and we can conclude the induction.

    \qed\end{proof}
Lemma~\ref{lem:travmon} immediately implies that the total cost of the fractional traversal algorithm given task sequence $\mathbf{\tilde{w}}$ is smaller than the cost of the fractional traversal algorithm given task sequence $\mathbf{w}$. This is since  $\delta_{\tilde{t}_0,\tilde{t}_i}\leq\delta_{t_0,t_i}$ for every $i\in\{0,\ldots,n\}$.

\subsection{``Follow the Traversal" Algorithm} \label{sec:algdesc}

We would like to come up with a dynamic pricing scheme for metrical task systems for which agents driven by self interest follow the  traversal sequence.

    We now describe an new algorithm  which can be simulated by a dynamic pricing scheme that encourages greedy agents to ``follow"
the fractional traversal algorithm. The algorithm, for a given metrical task system (a set of states $S$ and a distance metric $d:S\times S\mapsto \Reals^+$) and traversal sequence $\tau$, receives and processes a sequence of tasks $w^1,w^2,\ldots$ online.

    Roughly speaking, after processing each task, the algorithm simulates the traversal algorithm presented in Section \ref{sec:travdesc}. When handling task ${i}$, the algorithm views all states between its current position in the sequence, and the fractional traversal algorithm's current position in the sequence, $t_{i-1}$, as a superstate where the cost of a task in a superstate is the min cost over its constituent states.  The details of the algorithm are given in Algorithm~\ref{alg:ftt}.

\begin{alg}[tb]
\textbf{Input:} A traversal sequence $\tau$, a set of states $S$, a distance metric $d$ and an online sequence of tasks $w^1,w^2,\ldots$.

Let $\ell_{i-1}$ denote the position of the algorithm in the traversal sequence after completing tasks $w^1,\ldots,w^{i-1}$ and $t_{i-1}$ be the position of the traversal algorithm after completing these tasks ($\ell_0=t_0=1$).

Given task $w^{i}$:
\begin{enumerate}
    \item Let $j_{\mathrm{min}}\gets\min(\ell_{i-1}, t_{i-1})$ \\ and $j_{\mathrm{max}}\gets\max(\ell_{i-1}, t_{i-1})$.
    \item Let $D_0\gets\{{j_{\mathrm{min}}},\dots, {j_{\mathrm{max}}}\}$.
    \item Define $\tilde{c}(j) =
    \begin{cases}
    w^{i}_{\tau_j} &j\in D_0\\
    w^{i}_{\tau_j}+\delta_{j_{\mathrm{max}},j} &j> j_{\mathrm{max}}\\
    \infty & \mbox{otherwise}
    \end{cases}$
    \item Set $\ell_{i}\gets\argmin_{j\geq j_{\mathrm{min}}}{\tilde{c}({j})}$.
\end{enumerate}
\caption{``Follow the traversal" algorithm that allows going to all the states between its current state and the current state of the traversal algorithm ``for free".}\label{alg:ftt}
\end{alg}

Even though there may be several indices that minimize $\tilde{c}$, tie breaking can be done arbitrarily, since all the claims hold for every index which minimizes $\tilde{c}$.

\subsection{Proof of Approximation}
\begin{theorem} \label{thm:ft_approx}
    For any traversal sequence $\tau$, distance metric $d$, and sequence of tasks $w^1,\ldots, w^n$, let $C^{A2}=w^{A2}+d^{A2}$ be the total cost of Algorithm \ref{alg:ftt}, and $C^T=w^T+d^T$ the total cost of the traversal algorithm for that sequence of tasks (the processing cost plus the transition cost). Then $C^{A2}\leq 2\cdot C^T$.
\end{theorem}
\begin{proof}
    We use a potential argument. The potential function we use is $\Phi=\delta_{\ell_{i},t_{i}}$, which is the traversal distance the current position of Algorithm \ref{alg:ftt} and the current position of the traversal algorithm. We show that for every incoming task ${i}$,
    \begin{eqnarray*}
        C^{A2}_{i}+\Delta\Phi_{i}\leq 2\cdot C^T_{i},
    \end{eqnarray*}
    where $C^{A2}_{i}$ and $C^T_{i}$ are the costs of Algorithm \ref{alg:ftt} and Algorithm \ref{alg:t} respectively when processing the ${i}$th task, and $\Delta\Phi_{i}$ is the change in potential after processing it.

Let $T_{i} = \{ j : t_{i-1} \leq j \leq t_{i} \}$ be the set of positions that the traversal algorithm visits while processing the ${i}$-th task, and let $\lambda^{i}_j$ be the amount of work done in $\tau_j$ in the traversal while processing the task. Since the traversal algorithm processes the whole task, $\sum_{j\in T_{i}}\lambda^{i}_j=1$.

Note that $$C^{T}_{i} = \delta_{t_{i-1},t_{i}} + \sum_{j\in T_{i}} \lambda^{i}_j\cdot w^{i}_{\tau_j},$$
and also that
\begin{equation}
\label{eq:jobft}
C^{A2}_{i} \leq \delta_{\ell_{i-1}, \ell_{i}} + w^{i}_{\tau_{\ell_{i}}},
\end{equation}
since Algorithm \ref{alg:ftt} may use a direct edge from $\tau_{\ell_{i-1}}$ to $\tau_{\ell_{i}}$.
By definition, the difference in the potential after processing the ${i}$th task is:
\begin{equation}
\label{eq:potchange}
\Delta \Phi_{i} = -\delta_{\ell_{i-1},t_{i-1}} + \delta_{\ell_{i},t_{i}}
\end{equation}

Let $I_{j}$ be an indicator as to whether $\tau_j \notin D_0$.
Algorithm \ref{alg:ftt} chooses an index that minimizes $\tilde{c}$, therefore
$w^{i}_{\tau_{\ell_{i}}} + \delta_{j_{max},\ell_{i}}\cdot I_{\ell_{i}}
\leq w^{i}_{\tau_{j}} + \delta_{j_{max},j} \cdot I_{j}$.
Hence, for any $j \in \{t_{i-1}, \dots, t_{i}\}$  we have
$w^{i}_{\tau_{\ell_{i}}} + \delta_{j_{max},\ell_{i}}\cdot I_{\ell+1}
\leq w^{i}_{\tau_{j}} + \delta_{j_{max},t_{i}} \cdot I_{t_{i}}$.
Since $\lambda^j_{i}$ is a convex combination we have:
\begin{eqnarray}
\label{eq:foltrav}
& & w^{i}_{\tau_{\ell_{i}}} + \delta_{j_{max},\ell_{i}}\cdot I_{\ell+1}
\leq \\
& &\sum_{j \in \{t_{i-1},\dots t_{i}\}} \lambda^{i}_j \cdot\left( w^{i}_{\tau_{j}} + \delta_{j_{max},t_{i}} \cdot I_{t_{i}}\right) =\nonumber\\
& &\sum_{j \in \{t_{i-1},\dots t_{i}\}} \lambda^{i}_j \cdot w^{i}_{\tau_{j}} + \delta_{j_{max},t_{i}} \cdot I_{t_{i}} \leq \nonumber\\
& & C^T_{i}.\nonumber
\end{eqnarray}

\newcommand{\eqpc}{{\stackrel{(\ref{eq:potchange})}{=}}}
\newcommand{\leqpc}{{\stackrel{(\ref{eq:potchange})}{\leq}}}
\newcommand{\leqjt}{{\stackrel{(\ref{eq:jobft})}{\leq}}}
\newcommand{\leqft}{{\stackrel{(\ref{eq:foltrav})}{\leq}}}

\textbf{Case 1:} $\ell_{i-1}\leq t_{i-1}$ and $\ell_{i} \leq t_{i-1}$
\begin{figure}[H]
\scalebox{.75}{ \begin{tikzpicture}[line cap=round,line join=round,>=triangle 45,x=1.0cm,y=0.5cm]
\clip(-4.300000000000001,0.8600000000000003) rectangle (6,3.600000000000002);
\draw [line width=2.0pt,domain=-4.300000000000001:6] plot(\x,{(--7.7104-0.0*\x)/3.1599999999999997});
\begin{scriptsize}
\draw [fill=qqqqff] (-2.0,2.44) circle (4pt);
\draw[color=qqqqff] (-1.9,3.1) node {\Large$\ell_{i-1}$};
\draw [fill=qqqqff] (1.86,2.44) circle (4pt);
\draw[color=qqqqff] (1.96,3.1) node {\Large$\ell_{i}$};
\draw [fill=xdxdff] (-0.38,2.44) circle (3pt);
\draw [fill=dcrutc] (3.16,2.44) circle (4pt);
\draw[color=dcrutc] (3.26,3.1) node {\Large$t_{i-1}$};
\draw [fill=xdxdff] (-3.4,2.44) circle (3pt);
\draw [fill=xdxdff] (3.6,2.44) circle (3pt);
\draw [fill=xdxdff] (5.38,2.44) circle (3pt);
\draw [fill=dcrutc] (4.76,2.44) circle (4pt);
\draw[color=dcrutc] (4.86,3.1) node {\Large$t_{i}$};
\draw [fill=xdxdff] (0.4,2.44) circle (3pt);
\end{scriptsize}
\end{tikzpicture}}
\end{figure}
In this case, $\Delta \Phi = -\delta_{\ell_{i-1},\ell_{i}} + \delta_{t_{i-1},t_{i}}$. Therefore,
\begin{eqnarray*}
    C^{A2}_{i} + \Delta \Phi & \leqjt & w^{i}_{\tau_{\ell_{i}}} + \delta_{\ell_{i-1}, \ell_{i}} + \Delta \Phi
    \eqpc w^{i}_{\tau_{\ell_{i}}} + \delta_{t_{i-1}, t_{i}}
    \\ & \leqft & C^T_{i} + \delta_{t_{i-1},t_{i}} \leq 2 C^T_{i}.
\end{eqnarray*}
\textbf{Case 2:} $\ell_{i-1}\leq t_{i-1}$ and $t_{i-1} < \ell_{i} \leq t_{i}$
\begin{figure}[H]
\scalebox{.75}{ \begin{tikzpicture}[line cap=round,line join=round,>=triangle 45,x=1.0cm,y=0.5cm]
\clip(-4.300000000000001,0.8600000000000003) rectangle (6,3.600000000000002);
\draw [line width=2.0pt,domain=-4.300000000000001:6] plot(\x,{(--7.7104-0.0*\x)/3.1599999999999997});
\begin{scriptsize}
\draw [fill=qqqqff] (-2.0,2.44) circle (4pt);
\draw[color=qqqqff] (-1.9,3.1) node {\Large$\ell_{i-1}$};
\draw [fill=dcrutc] (1.86,2.44) circle (4pt);
\draw[color=dcrutc] (1.96,3.1) node {\Large$t_{i-1}$};
\draw [fill=xdxdff] (-0.38,2.44) circle (3pt);
\draw [fill=qqqqff] (3.16,2.44) circle (4pt);
\draw[color=qqqqff] (3.26,3.1) node {\Large$\ell_{i}$};
\draw [fill=xdxdff] (-3.4,2.44) circle (3pt);
\draw [fill=xdxdff] (3.6,2.44) circle (3pt);
\draw [fill=xdxdff] (5.38,2.44) circle (3pt);
\draw [fill=dcrutc] (4.76,2.44) circle (4pt);
\draw[color=dcrutc] (4.86,3.1) node {\Large$t_{i}$};
\draw [fill=xdxdff] (0.4,2.44) circle (3pt);
\end{scriptsize}
\end{tikzpicture}}
\end{figure}
We have,
\begin{eqnarray*}
C^{A2}_{i} + \Delta \Phi &\leqjt& w^{i}_{\tau_{\ell_{i}}} + \delta_{\ell_{i-1}, t_{i-1}} + \delta_{t_{i-1}, \ell_{i}} + \Delta \Phi \\
&\eqpc&  w^{i}_{\tau_{\ell_{i}}} + \delta_{t_{i-1}, \ell_{i}} + \delta_{\ell_{i},t_{i}}\\
& \leqft & C^T_{i} + \delta_{\ell_{i},t_{i}}\\
&  \leq & 2 C^T_{i}.
\end{eqnarray*}
\textbf{Case 3:} $\ell_{i-1}\leq t_{i-1}$ and $t_{i} < \ell_{i}$
\begin{figure}[H]
\scalebox{.75}{ \begin{tikzpicture}[line cap=round,line join=round,>=triangle 45,x=1.0cm,y=0.5cm]
\clip(-4.300000000000001,0.8600000000000003) rectangle (6,3.600000000000002);
\draw [line width=2.0pt,domain=-4.300000000000001:6] plot(\x,{(--7.7104-0.0*\x)/3.1599999999999997});
\begin{scriptsize}
\draw [fill=qqqqff] (-2.0,2.44) circle (4pt);
\draw[color=qqqqff] (-1.9,3.1) node {\Large$\ell_{i-1}$};
\draw [fill=dcrutc] (1.86,2.44) circle (4pt);
\draw[color=dcrutc] (1.96,3.1) node {\Large$t_{i-1}$};
\draw [fill=xdxdff] (-0.38,2.44) circle (3pt);
\draw [fill=dcrutc] (3.16,2.44) circle (4pt);
\draw[color=dcrutc] (3.26,3.1) node {\Large$t_{i}$};
\draw [fill=xdxdff] (-3.4,2.44) circle (3pt);
\draw [fill=xdxdff] (3.6,2.44) circle (3pt);
\draw [fill=xdxdff] (5.38,2.44) circle (3pt);
\draw [fill=qqqqff] (4.76,2.44) circle (4pt);
\draw[color=qqqqff] (4.86,3.1) node {\Large$\ell_{i}$};
\draw [fill=xdxdff] (0.4,2.44) circle (3pt);
\end{scriptsize}
\end{tikzpicture}}
\end{figure}
\begin{eqnarray*}
    C^{A2}_{i} + \Delta \Phi &\leqjt& w^{i}_{\tau_{\ell_{i}}} + \delta_{\ell_{i-1}, t_{i-1}} + \delta_{t_{i-1}, \ell_{i}} + \Delta \Phi \\
    & \eqpc & w^{i}_{\tau_{\ell_{i}}} + \delta_{t_{i-1}, \ell_{i}} + \delta_{t_{i},\ell_{i}} \\
    & \leqft & C^T_{i} + \delta_{t_{i},\ell_{i}} \\
    & \leq & 2 C^T_{i}.
\end{eqnarray*}
\textbf{Case 4:} $t_{i-1} < \ell_{i-1}$ and $t_{i} \leq \ell_{i} \leq \ell_{i-1}$
\begin{figure}[H]
\scalebox{.75}{ \begin{tikzpicture}[line cap=round,line join=round,>=triangle 45,x=1.0cm,y=0.5cm]
\clip(-4.300000000000001,0.8600000000000003) rectangle (6,3.600000000000002);
\draw [line width=2.0pt,domain=-4.300000000000001:6] plot(\x,{(--7.7104-0.0*\x)/3.1599999999999997});
\begin{scriptsize}
\draw [fill=dcrutc] (-2.0,2.44) circle (4pt);
\draw[color=dcrutc] (-1.9,3.1) node {\Large$t_{i-1}$};
\draw [fill=dcrutc] (1.86,2.44) circle (4pt);
\draw[color=dcrutc] (1.96,3.1) node {\Large$t_{i}$};
\draw [fill=xdxdff] (-0.38,2.44) circle (3pt);
\draw [fill=qqqqff] (3.16,2.44) circle (4pt);
\draw[color=qqqqff] (3.26,3.1) node {\Large$\ell_{i}$};
\draw [fill=xdxdff] (-3.4,2.44) circle (3pt);
\draw [fill=xdxdff] (3.6,2.44) circle (3pt);
\draw [fill=xdxdff] (5.38,2.44) circle (3pt);
\draw [fill=qqqqff] (4.76,2.44) circle (4pt);
\draw[color=qqqqff] (4.86,3.1) node {\Large$\ell_{i-1}$};
\draw [fill=xdxdff] (0.4,2.44) circle (3pt);
\end{scriptsize}
\end{tikzpicture}}
\end{figure}
In this case, $\Delta \Phi = -\delta_{t_{i-1},t_{i}} - \delta_{\ell_{i},\ell_{i-1}}$. Hence,

$$C^{A2}_{i} + \Delta \Phi \leqjt w^{i}_{\tau_{\ell_{i}}} + \delta_{\ell_{i}, \ell_{i-1}} + \Delta \Phi \eqpc  w^{i}_{\tau_{\ell_{i}}}  - \delta_{t_{i-1},t_{i}} \leqft C^T_{i}.$$
\textbf{Case 5:} $t_{i-1} < \ell_{i-1}$ and $t_{i} \leq \ell_{i-1} < \ell_{i} $
\begin{figure}[H]
\scalebox{.75}{ \begin{tikzpicture}[line cap=round,line join=round,>=triangle 45,x=1.0cm,y=0.5cm]
\clip(-4.300000000000001,0.8600000000000003) rectangle (6,3.600000000000002);
\draw [line width=2.0pt,domain=-4.300000000000001:6] plot(\x,{(--7.7104-0.0*\x)/3.1599999999999997});
\begin{scriptsize}
\draw [fill=dcrutc] (-2.0,2.44) circle (4pt);
\draw[color=dcrutc] (-1.9,3.1) node {\Large$t_{i-1}$};
\draw [fill=dcrutc] (1.86,2.44) circle (4pt);
\draw[color=dcrutc] (1.96,3.1) node {\Large$t_{i}$};
\draw [fill=xdxdff] (-0.38,2.44) circle (3pt);
\draw [fill=qqqqff] (3.16,2.44) circle (4pt);
\draw[color=qqqqff] (3.26,3.1) node {\Large$\ell_{i-1}$};
\draw [fill=xdxdff] (-3.4,2.44) circle (3pt);
\draw [fill=xdxdff] (3.6,2.44) circle (3pt);
\draw [fill=xdxdff] (5.38,2.44) circle (3pt);
\draw [fill=qqqqff] (4.76,2.44) circle (4pt);
\draw[color=qqqqff] (4.86,3.1) node {\Large$\ell_{i}$};
\draw [fill=xdxdff] (0.4,2.44) circle (3pt);
\end{scriptsize}
\end{tikzpicture}}
\end{figure}
In this case, $\Delta \Phi = -\delta_{t_{i-1},t_{i}} + \delta_{\ell_{i-1},\ell_{i}}$.
Therefore,
$$C^{A2}_{i} + \Delta \Phi \leqjt w^{i}_{\tau_{\ell_{i}}} + \delta_{\ell_{i-1}, \ell_{i}} + \Delta \Phi \leqpc  w^{i}_{\tau_{\ell_{i}}} + 2 \delta_{\ell_{i-1}, \ell_{i}} \leqft 2 C^T_{i}.$$
\textbf{Case 6:} $t_{i-1} < \ell_{i-1}$ and $\ell_{i} < t_{i} \leq \ell_{i-1}$
\begin{figure}[H]
\scalebox{.75}{ \begin{tikzpicture}[line cap=round,line join=round,>=triangle 45,x=1.0cm,y=0.5cm]
\clip(-4.300000000000001,0.8600000000000003) rectangle (6,3.600000000000002);
\draw [line width=2.0pt,domain=-4.300000000000001:6] plot(\x,{(--7.7104-0.0*\x)/3.1599999999999997});
\begin{scriptsize}
\draw [fill=dcrutc] (-2.0,2.44) circle (4pt);
\draw[color=dcrutc] (-1.9,3.1) node {\Large$t_{i-1}$};
\draw [fill=qqqqff] (1.86,2.44) circle (4pt);
\draw[color=qqqqff] (1.96,3.1) node {\Large$\ell_{i}$};
\draw [fill=xdxdff] (-0.38,2.44) circle (3pt);
\draw [fill=dcrutc] (3.16,2.44) circle (4pt);
\draw[color=dcrutc] (3.26,3.1) node {\Large$t_{i}$};
\draw [fill=xdxdff] (-3.4,2.44) circle (3pt);
\draw [fill=xdxdff] (3.6,2.44) circle (3pt);
\draw [fill=xdxdff] (5.38,2.44) circle (3pt);
\draw [fill=qqqqff] (4.76,2.44) circle (4pt);
\draw[color=qqqqff] (4.86,3.1) node {\Large$\ell_{i-1}$};
\draw [fill=xdxdff] (0.4,2.44) circle (3pt);
\end{scriptsize}
\end{tikzpicture}}
\end{figure}
In this case, $\Delta \Phi = -\delta_{t_{i-1},\ell_{i}} - \delta_{t_{i},\ell_{i-1}}$. Therefore,
\begin{eqnarray*}
C^{A2}_{i} + \Delta \Phi &\leqjt& w^{i}_{\tau_{\ell_{i}}} + \delta_{\ell_{i}, \ell_{i-1}} + \Delta \Phi \\
& = & w^{i}_{\tau_{\ell_{i}}}  + \delta_{\ell_{i},t_{i}} + \delta_{t_{i},\ell_{i-1}} + \Delta \Phi \\
&\leqpc& w^{i}_{\tau_{\ell_{i}}} + \delta_{\ell_{i},t_{i}}  \\
& \leqft& 2 C^T_{i}.
\end{eqnarray*}
\textbf{Case 7:} $t_{i-1} < \ell_{i-1}$ and $\ell_{i} < \ell_{i-1} < t_{i}$
\begin{figure}[H]
\scalebox{.75}{ \begin{tikzpicture}[line cap=round,line join=round,>=triangle 45,x=1.0cm,y=0.5cm]
\clip(-4.300000000000001,0.8600000000000003) rectangle (6,3.600000000000002);
\draw [line width=2.0pt,domain=-4.300000000000001:6] plot(\x,{(--7.7104-0.0*\x)/3.1599999999999997});
\begin{scriptsize}
\draw [fill=dcrutc] (-2.0,2.44) circle (4pt);
\draw[color=dcrutc] (-1.9,3.1) node {\Large$t_{i-1}$};
\draw [fill=qqqqff] (1.86,2.44) circle (4pt);
\draw[color=qqqqff] (1.96,3.1) node {\Large$\ell_{i}$};
\draw [fill=xdxdff] (-0.38,2.44) circle (3pt);
\draw [fill=qqqqff] (3.16,2.44) circle (4pt);
\draw[color=qqqqff] (3.26,3.1) node {\Large$\ell_{i-1}$};
\draw [fill=xdxdff] (-3.4,2.44) circle (3pt);
\draw [fill=xdxdff] (3.6,2.44) circle (3pt);
\draw [fill=xdxdff] (5.38,2.44) circle (3pt);
\draw [fill=dcrutc] (4.76,2.44) circle (4pt);
\draw[color=dcrutc] (4.86,3.1) node {\Large$t_{i}$};
\draw [fill=xdxdff] (0.4,2.44) circle (3pt);
\end{scriptsize}
\end{tikzpicture}}
\end{figure}
In this case, $\Delta \Phi = -\delta_{\ell_{i},\ell_{i-1}} - \delta_{t_{i-1},\ell_{i}}
+ \delta_{\ell_{i},t_{i}} \leq -\delta_{\ell_{i},\ell_{i-1}} + \delta_{t_{i-1},t_{i}}$. Therefore,

\begin{align*}
C^{A2}_{i} + \Delta \Phi &\leqjt w^{i}_{\tau_{\ell_{i}}} + \delta_{\ell_{i}, \ell_{i-1}} + \Delta \Phi \leqpc  w^{i}_{\tau_{\ell_{i}}}  +  \delta_{t_{i-1},t_{i}} \\
&\leqft 2 C^T_{i} \enspace.
\end{align*}

\textbf{Case 8:} $t_{i-1} < \ell_{i-1}$ and $\ell_{i-1} \leq \ell_{i}  < t_{i}$
\begin{figure}[H]
\scalebox{.75}{ \begin{tikzpicture}[line cap=round,line join=round,>=triangle 45,x=1.0cm,y=0.5cm]
\clip(-4.300000000000001,0.8600000000000003) rectangle (6,3.600000000000002);
\draw [line width=2.0pt,domain=-4.300000000000001:6] plot(\x,{(--7.7104-0.0*\x)/3.1599999999999997});
\begin{scriptsize}
\draw [fill=dcrutc] (-2.0,2.44) circle (4pt);
\draw[color=dcrutc] (-1.9,3.1) node {\Large$t_{i-1}$};
\draw [fill=qqqqff] (1.86,2.44) circle (4pt);
\draw[color=qqqqff] (1.96,3.1) node {\Large$\ell_{i-1}$};
\draw [fill=xdxdff] (-0.38,2.44) circle (3pt);
\draw [fill=qqqqff] (3.16,2.44) circle (4pt);
\draw[color=qqqqff] (3.26,3.1) node {\Large$\ell_{i}$};
\draw [fill=xdxdff] (-3.4,2.44) circle (3pt);
\draw [fill=xdxdff] (3.6,2.44) circle (3pt);
\draw [fill=xdxdff] (5.38,2.44) circle (3pt);
\draw [fill=dcrutc] (4.76,2.44) circle (4pt);
\draw[color=dcrutc] (4.86,3.1) node {\Large$t_{i}$};
\draw [fill=xdxdff] (0.4,2.44) circle (3pt);
\end{scriptsize}
\end{tikzpicture}}
\end{figure}
In this case, $\Delta \Phi = -\delta_{t_{i-1},\ell_{i-1}} + \delta_{\ell_{i},t_{i}}
\leq \delta_{t_{i-1},t_{i}}$. Therefore,
\begin{eqnarray*}
    C^{A2}_{i} + \Delta \Phi &\leqjt& w^{i}_{\tau_{\ell_{i}}} + \delta_{\ell_{i-1}, \ell_{i}} + \Delta \Phi \\
    &\leqpc&  w^{i}_{\tau_{\ell_{i}}} + \delta_{\ell_{i-1}, \ell_{i}} +  \delta_{t_{i-1},t_{i}}\\
    & \leqft& 2 C^T_{i}.
\end{eqnarray*}
\textbf{Case 9:} $t_{i-1} < \ell_{i-1}$ and $\ell_{i-1} \leq t_{i} \leq \ell_{i}$
\begin{figure}[H]
\scalebox{.75}{ \begin{tikzpicture}[line cap=round,line join=round,>=triangle 45,x=1.0cm,y=0.5cm]
\clip(-4.300000000000001,0.8600000000000003) rectangle (6,3.600000000000002);
\draw [line width=2.0pt,domain=-4.300000000000001:6] plot(\x,{(--7.7104-0.0*\x)/3.1599999999999997});
\begin{scriptsize}
\draw [fill=dcrutc] (-2.0,2.44) circle (4pt);
\draw[color=dcrutc] (-1.9,3.1) node {\Large$t_{i-1}$};
\draw [fill=qqqqff] (1.86,2.44) circle (4pt);
\draw[color=qqqqff] (1.96,3.1) node {\Large$\ell_{i-1}$};
\draw [fill=xdxdff] (-0.38,2.44) circle (3pt);
\draw [fill=dcrutc] (3.16,2.44) circle (4pt);
\draw[color=dcrutc] (3.26,3.1) node {\Large$t_{i}$};
\draw [fill=xdxdff] (-3.4,2.44) circle (3pt);
\draw [fill=xdxdff] (3.6,2.44) circle (3pt);
\draw [fill=xdxdff] (5.38,2.44) circle (3pt);
\draw [fill=qqqqff] (4.76,2.44) circle (4pt);
\draw[color=qqqqff] (4.86,3.1) node {\Large$\ell_{i}$};
\draw [fill=xdxdff] (0.4,2.44) circle (3pt);
\end{scriptsize}
\end{tikzpicture}}
\end{figure}
In this case, $\Delta \Phi = -\delta_{t_{i-1},\ell_{i-1}} + \delta_{t_{i},\ell_{i}}
\leq \delta_{\ell_{i-1},\ell_{i}}$. Therefore,
\begin{eqnarray*}
C^{A2}_{i} + \Delta \Phi & \leqjt & w^{i}_{\tau_{\ell_{i}}} + \delta_{\ell_{i-1}, \ell_{i}} + \Delta \Phi\\
& \leqpc & w^{i}_{\tau_{\ell_{i}}} + 2 \delta_{\ell_{i-1}, \ell_{i}} \\
& \leqft & 2 C^T_{i}.
\end{eqnarray*}
\qed\end{proof}

\subsection{Pricing Scheme} \label{sec:mtsprice}
In this section, we define a dynamic pricing scheme that observes decisions made by greedy agents and sets a surcharge for the next agent to arrive.
We want to set prices so that states chosen by (greedy) agents are the same as the ones chosen by Algorithm~\ref{alg:ftt} (up to tie breaking) given the sequence of tasks $\mathbf{w}$.

The problem is that dynamic pricing schemes, as we've defined them, can only observe the greedy behavior of agents. the state chosen by the agents, and the amount of work done in that state while processing the task; Namely, after agent ${i}$ has selfishly used the system, the mechanism observes $s_{i}$, the state chosen by agent ${i}$, and $w^i_{s_{i}}$, the amount of work done in $s_{i}$ while processing the task, but the other coordinates of $w^{i}$ $(w^i_1,\ldots,w^i_{s_{i}-1},w^i_{s_{i}+1},\ldots,w^i_m)$ are not known to it.

We resolve this issue by giving the pricing scheme a different (imaginary) sequence of tasks $\mathbf{\tilde{w}}$ with the following \emph{consistency} properties:
\begin{enumerate}
\item For every ${i}$, $w^i_{s_{i}}$ is equal to $\tilde{w}^i_{s_{i}}$.
\item For every ${i}$ and for every $j\in S\setminus\{s_{i}\}$, $w^i_j\geq \tilde{w}^i_j$.
\item For every ${i}$, given the pricing function $P_{i}$ over the state of the system after observing events $1,\ldots, {i}-1$, the computed vector $\tilde{w}^{i}$ is consistent with the agent selfishly choosing to process her task in state $s_{i}$; I.e., for all $j$
\end{enumerate}
\begin{equation}\label{prop3}
\tilde{w}^i_j+d_{s_{i-1},j}+P_{i}(j)\geq {w}^i_{s_{i}}+d_{s_{i-1},s_i}+P_{i}(s_{i}) \enspace.
\end{equation}

Note that running Algorithm~\ref{alg:ftt} on the sequence of tasks $\mathbf{\tilde{w}}$ yields a 2-approximation to the cost of the Fractional Traversal Algorithm on the sequence of tasks $\mathbf{w}$: by Lemma~\ref{lem:travmon}, the cost of the Fractional Traversal Algorithm on $\mathbf{\tilde{w}}$ is no larger than its cost on $\mathbf{w}$, and by Theorem~\ref{thm:ft_approx}, Algorithm~\ref{alg:ftt} gives a 2-approximation to the cost of the Fractional Traversal Algorithm (on any sequence, in particular $\mathbf{\tilde{w}}$).  Therefore, it suffices to find a pricing scheme that makes the agents follow Algorithm~\ref{alg:ftt} on sequence $\mathbf{\tilde{w}}$.

We show how to find $\mathbf{\tilde{w}}$ that satisfies the consistency properties in Section~\ref{sec:lbdesc}. From now on, we assume such $\tilde{w}^{i-1}$ can be computed after observing agent ${i}$'s greedy behavior. We show how to devise a pricing scheme that given $\mathbf{\tilde{w}}$, makes incoming agents act \emph{as though} they are being guided by Algorithm \ref{alg:ftt} given input $\mathbf{\tilde{w}}$. We describe how to do this in Pricing Scheme \ref{pricing:ftt}. We note that although Pricing Scheme \ref{pricing:ftt} may price states using negative number, we can increase all prices by a large enough constant without affecting the decisions of the agents.

\begin{pricing}[tb]
\textbf{Input:} A traversal sequence $\tau$, a set of states $S$, a distance metric $d$ and an observed sequence of decisions $(s_1,w^1_{s_1}), (s_2,w^2_{s_2}), \ldots, (s_{i-1},w^i_{s_{i-1}})$,  made by agents $1, 2, \ldots, {i-1}$.

            The variables $\ell_{i-1}$ and $t_{i-1}$ represent two indices in the traversal sequence, set after the ${i-1}$th agent's decision. For ${i-1}=0$ take $\ell_0\gets 1$ and $t_0\gets 1$. \\

            \textbf{Setting prices for agent ${i}$:}
            \begin{enumerate}
                \item Let $j_{\mathrm{min}}\gets\min(\ell_{i-1}, t_{i-1})$ \\ and $j_{\mathrm{max}}\gets\max(\ell_{i-1}, t_{i-1})$.
                \item Let $D_0\gets\{{j_{\mathrm{min}}},\dots, {j_{\mathrm{max}}}\}$.
                \item For a state $s$ let $m(s) \gets \min\{j\geq \jmin| \tau_j = s \}$.
                \item Define $P_{i}(s)\equiv
                \begin{cases}
                -d_{s_{i-1},s} & m(s) \in D_0,\\
                -d_{s_{i-1},s} + \delta_{j_{\mathrm{max}}, m(s)} &\mbox{otherwise.}\\
                \end{cases}$
            \end{enumerate}
            \textbf{Determining $\ell_{i}$ and $t_{i}$ after observing $(s_{i},w^{i}_{s_{i}})$:}
            \begin{enumerate}
                \item Set $\ell_{i}\gets m(s_{i})$.
                \item Compute $\tilde{w}^{i}$ using the set of equations (\ref{eq:lbtask}) as described in Section \ref{sec:lbdesc}.
                \item Update $t_{i}$ as done by the Fractional Traversal Algorithm (Algorithm \ref{alg:t}) using $\tilde{w}^{i}$ as input.
            \end{enumerate}
        \caption{A dynamic pricing scheme that makes incoming agents ``follow" Algorithm \ref{alg:ftt}.}\label{pricing:ftt}
    \end{pricing}

    \begin{theorem}
        Pricing Scheme \ref{pricing:ftt} makes greedy agents simulate Algorithm \ref{alg:ftt} with task sequence $\mathbf{\tilde{w}}$ given as input.
    \end{theorem}
    \begin{proof}
        For every state $s'\in S$, let $m(s')$ be as defined in Pricing Scheme \ref{pricing:ftt} and let $\tilde{c}$ be as defined in Algorithm \ref{alg:ftt}. Since $m(s')\in \argmin_{j:\tau_j=s'}\tilde{c}(j)$, we may assume that Algorithm \ref{alg:ftt} always prefers $m(s')$ amongst the indices for which $\tau_j=s'$.

        Let $s$ be the state chosen by the agent.
         If $m(s)\in D_0$, then the disutility of the agent is $w^{i}_s + d_{s_{i-1},s} + P_{i}(s) = w^{i}_s+ d_{s_{i-1},s} - d_{s_{i-1},s} = \tilde{c}(m(s))$. If $m(s)\notin D_0$, then now the disutility is $w^{i}_s + d_{s_{i-1},s} + P_{i}(s) = w^{i}_s + d_{s_{i-1},s} - d_{s_{i-1},s} + \delta_{j_{\mathrm{max}},m(s)} = \tilde{c}({m(s)})$.

        By the way $\tilde{w}^{i}$ is computed, we get that if the algorithm would have received $\tilde{w}^{i}$ as input, then given $\tilde{s}\in S$ that was not chosen by agent ${i}$, we have that,
        \begin{itemize}
            \item if $m(\tilde{s})\in D_0$, then
            \begin{eqnarray*}
            \tilde{c}({m(\tilde{s})}) & = & \tilde{w}^{i}_{\tilde{s}}+ d_{s_{i-1},\tilde{s}} - d_{s_{i-1},\tilde{s}}\\
            & = & \tilde{w}^{i}_{\tilde{s}}+ d_{s_{i-1},\tilde{s}} + P_{i}(\tilde{s})\\
            &\geq& w^{i}_{s}+ d_{s_{i-1},s} + P_{i}(s)\\
            & = & \tilde{c}({m(s)}),
            \end{eqnarray*}
            where the third inequality is derived from (\ref{prop3}), and
            \item if $m(\tilde{s})\notin D_0$, then
            \begin{eqnarray*}
            \tilde{c}({m(\tilde{s})})& = &\tilde{w}^{i}_{\tilde{s}} + d_{s_{i-1},\tilde{s}} - d_{s_{i-1},\tilde{s}} + \delta_{j_{\mathrm{max}},m(\tilde{s})}\\
            &=& \tilde{w}^{i}_{\tilde{s}} + d_{s_{i-1},\tilde{s}} + P_{i}(\tilde{s})\\
            &\geq& w^{i}_{s} + d_{s_{i-1},s} + P_{i}(s) \\
            &=& \tilde{c}({m(s)}),
            \end{eqnarray*}
            where the third inequality is derived from (\ref{prop3}).
        \end{itemize}

        Therefore, the agent picks a state $s$ such that $m(s)$ minimizes $\tilde{c}$ given task $\tilde{w}^{i}$. Algorithm \ref{alg:ftt} also chooses an index that minimizes $\tilde{c}$. Since the tie breaking of Algorithm \ref{alg:ftt} can be done arbitrarily, we can assume that the algorithm chooses index $m(s)$ as well when given task $\tilde{w}^{i}$.
    \qed\end{proof}
    \begin{corollary}
        Pricing Scheme \ref{pricing:ftt} yields a price of anarchy of $16(m-1)$.
    \end{corollary}

\subsection{Computing ``Imaginary" Tasks}\label{sec:lbdesc}

As mentioned in Section \ref{sec:mtsprice}, Pricing Scheme \ref{pricing:ftt} cannot observe the real sequence of task vectors $\mathbf{w}$. This is problematic, since the pricing scheme runs a simulation of the fractional traversal algorithm, which takes the task vectors as input.

To resolve this issue, we need to compute imaginary task vectors with the three \emph{consistency} properties, as described in Section \ref{sec:mtsprice}. We are given a pricing function $P_{i}:S\mapsto \Reals^+$ that is set on states of the task system after the arrival of the ${i-1}$-th agent, a state $s_{i}$ chosen by the ${i}$-th agent, and the amount of work done in that state while processing the tasks (which is exactly $w^{i}_{s_{i}}$). Let $w^{i}=\left(w^{i}_1,w^{i}_2,\ldots,w^{i}_m\right)$ be the current incoming task (which we cannot observe). We claim that the vector $\tilde{w}^{i}$ defined below is coordinate-wise dominated by $w^{i}$ and consistent with the agent greedily choosing to process the task in $s_{i}$ while the states are priced using $P_{i}$:
\begin{align}\label{eq:lbtask}
    {\tilde{w}}^{i}_{j} := &\max\{0,\\ &w^{i}_{s_{i}}+d_{s_{i-1},s_{i}}+P_{i}\left(s_{i}\right)-d_{s_{i-1},j}-P_{i}(j)\}\nonumber
\end{align}
for all $j \in S$.  Note that ${\tilde{w}}^{i}_{s_{i}}={{w}}^{i}_{s_{i}}$.

By definition, ${w}^{i}_j \geq 0$, and since the agent chooses $s_{i}$, we have that ${w}^{i}_j+d_{s_{i-1},j}+P_{i}(j)\geq w^{i}_{s_{i}}+d_{s_{i-1},s_{i}}+P_{i}\left(s_{i}\right)$
, i.e.,
${w}^{i}_j \geq w^{i}_{s_{i}}+d_{s_{i-1},s_{i}}+P_{i}\left(s_{i}\right)-d_{s_{i-1},j}-P_{i}(j)$.
Therefore, it holds that ${w}^{i}_j \geq \tilde{w}^{i}_j$.
In addition,  since $\tilde{w}^{i}_{j}\geq w^{i}_{s_{i}}+d_{s_{i-1},s_{i}}+P_{i}\left(s_{i}\right)-d_{s_{i-1},j}-P_{i}(j)$, we have that $\tilde{w}^{i}_{j}+d_{s_{i-1},j}+P_{i}(j)\geq w^{i}_{s_{i}}+d_{s_{i-1},s_{i}}+P_{i}\left(s_{i}\right)$ as required in (\ref{prop3}). We conclude that we are able to compute a task vector which satisfies the desired \emph{consistency} properties, thus enabling us to devise pricing over states that insures a nearly optimal price of anarchy.

\section{Pricing Servers on Trees}
\renewcommand{\cal}[1]{\mathcal{#1}}
\newcommand{\loc}{local}
\newcommand{\laz}{lazy}
\newcommand{\m}{monotone}
\newcommand{\DC}{DC}
\newcommand{\St}[1]{S_{#1}}
\newcommand{\srv}{s}
\newcommand{\M}{\mathcal{M}}


The $k$-server problem,
introduced by Manasse~et~al.~\cite{ManasseMS90},
is defined as follows.  Let $(V,d)$ be a fixed metric space.
There are $k$ servers, initially located in specified points $S^0 = \{ s^0_1,s^0_2,\ldots,s^0_k \}$ s.t. for  $1\leq j \leq k$, $s^0_j \in V$, that can be moved
by the algorithm.  Moving a server from $x$ to $y$ costs $d(x,y)$.
The goal is to serve a sequence of arriving requests with minimum possible cost.
Each request $r_i$ is a point in $V$, and to serve it, the algorithm has to move (at least) one server to that point.  

In our setting, we assume that the requests are made by successive agents.
Moreover, we assume that each agent is greedy, and will therefore ask that
their request is served with the server that minimizes the sum of its
distance to the request and the price we set for the server. Therefore,
before the arrival of the $i$th agent, the algorithm sets a payment function for the servers: $P_i\colon S^{i-1} \mapsto \Reals $, and agent $i$ would choose to move server:
$$ \ell_i \in \argmin_j \{ d(r_i,s^{i-1}_j) + P_i(s^{i-1}_j)\},$$
and the next state would be $$S^{i} = S^{i-1}\setminus \{s^{i-1}_{\ell_i}\}\cup \{r_i\}.$$
The social cost is defined as:
$$f(\overline{\sigma},\overline{s}) = \sum_{i=1}^{|\overline{\sigma}|}  d(s^{i-1}_{\ell_i},r_i).$$

On a high level, we want to find out if any algorithm can be simulated by appropriate pricing,
i.e., whether given an algorithm $\A$ and its current state (locations of servers),
we can give a pricing of servers that will make an arriving greedy agent with a request $r$,
serve $r$ in the same way that $\A$ would have, wherever the request occurs.

For some algorithms, designing a payment scheme is straightforward.  For instance,
Irani and Rubinfeld proved that the following $2$-server algorithm \emph{Balance2}
is $10$-competitive on any metric space~\cite{IraniR91}.
The algorithm keeps track of the total distance that each of the two servers has traveled,
and serves each request with the server that minimizes the sum of its distance to the request
and \emph{half} the total distance it has traveled before.
Clearly, setting the price of each server to the second summand, i.e.,
half the total distance the server has moved already, will make the agents
follow the algorithm.  The only catch is that we have no control over tie-breaking
since it is done by the agents.  However, this can only affect the total cost
by an additive constant.

In general, however, simulating any algorithm $\A$ is not as easy.
For instance, $\A$ could move many servers (in a non-lazy fashion) upon a request.
This is the case with some known algorithms, for example with the
\emph{Double Coverage} ({\DC}) algorithm~\cite{ChrobakKPV91}
that is $k$-competitive for the line metric.
This algorithm has been generalized to all tree metrics,
retaining the optimal ratio~\cite{ChrobakL91}.
As an agent is only allowed to pick a server that will serve the request,
such complex algorithms cannot be simulated directly.

In this work, we restrict our attention to $k$-server in tree metrics,
and the optimal algorithm for this setting~\cite{ChrobakL91} in particular.
We observe that the standard transformation to make the algorithm {\em \laz},
if done carefully, preserves its other properties, which we call {\em locality}
and {\em monotonicity} (all defined later), and that combined these
enable simulation by a pricing scheme.
We remark that, in general, such transformation employs keeping track
of ``virtual positions'' of servers and basing decisions on those.
This makes designing a~pricing scheme to simulate the algorithm challenging,
since an agent's decision is based on the ``real positions'' of the servers
(and their prices).

\subsection{Notation}
Given an algorithm $\A$, let $\St{\A} = \{\srv_1,\ldots,\srv_k\}$ denote its \emph{state}, i.e.,
the set of points where $\A$'s servers are currently located.  We will associate the servers
with the points they are in.  Similarly, we will associate a request $r$
with the point in which it occurs.  Given two points $v$ and $w$ in a tree,
we let $\cal{P}(v,w)$ denote the unique path between $v$ and $w$.
We say that a server $\srv\in\St{\A}$ is \emph{adjacent} for a point $v$
if $\cal{P}(v,\srv)\cap\St{\A}=\{\srv\}$, i.e., there is no other server from $\St{\A}$
on the path $\cal{P}(v,\srv)$.

A deterministic online algorithm is \emph{\laz} if the only moves it makes consist in
moving a single server to the current request.  An algorithm is \emph{\loc} if
it always serves a request with a server that is adjacent to it.
Furthermore, an algorithm $\A$ is \emph{\m} if it satisfies the following property:
given $\St{\A}$ and a request at $r$, if the algorithm would serve a request at $r$ with
$\srv\in\St{\A}$, then it would also serve a request at any vertex on the path $\cal{P}(r,\srv)$
with $\srv$.

\subsection{Matchings}
Later, we give a transformation that takes an algorithm $\A$ and returns an
``equivalent'' algorithm $\A'$ that is {\laz}.  This transformation relies
on minimum cost (perfect, bipartite) matching of $\St{\A}$ to $\St{\A'}$,
where the costs are specified by the metric distance $d$.
Given two sets of points $X$ and $Y$ of equal cardinality, we denote any
such matching by $M(X,Y)$ and its cost by $d(X,Y)$.

When the metric space is a line, we associate points with real numbers,
use inequalities to specify the order of points, and note that
$d(x,y)=|x-y|$.
Moreover, we remark that in such case there exists a perfect matching
of particular structure.  Namely, given $X=\{x_1,x_2,\ldots,x_k\}$
and $Y=\{y_1,y_2,\ldots,y_k\}$ that are sorted from left to right,
i.e., $x_1 \leq x_2  \leq \ldots \leq x_k$ and $y_1 \leq y_2  \leq \ldots \leq y_k$,
we call a matching in which $x_i$ is matched to $y_i$ for all $i$
a \emph{canonical matching} and denote it $\M(X,Y)$.

\begin{lemma}\label{lem:line-match}
The canonical matching on a line is a minimum cost perfect matching.
\end{lemma}
\begin{proof}
Let $M=M(X,Y)$ be any minimum cost perfect matching of $X$ to $Y$.
We show how $M$ can be transformed to the canonical matching without increasing its cost.
As long as $M$ is not canonical, find the minimum $i$ such that $x_i$ is not matched to $y_i$.
Then $x_i$ is matched to some $y_j$ and $y_i$ is matched to some $x_h$.
Moreover, $j>i$ and $h>i$ since $i$ is minimal.
We assume wlog that $x_i \leq y_i$  (if not, swap the sets).
This implies that $y_j \geq y_i \geq x_i$.
Modify $M$ as follows: match $x_i$ with $y_i$ and $x_h$ with $y_j$.
Matching $x_i$ to $y_i$ decreases the total distance by $y_j-y_i$.
By the triangle inequality, $|x_h-y_j| \leq |x_h-y_i| + |y_i-y_j| =  |x_h-y_i| + y_j-y_i$,
i.e., matching  $x_h$ to $y_j$ increases the total distance by at most $y_j-y_i$.
Hence, the cost does not increase throughout the transformation,
and is thus minimized by the canonical matching $\M(X,Y)$ we obtain at the end.
\qed\end{proof}

Given $X$ and $Y$ and a point $r\in X$, we call a matching of $X$ to $Y$ \emph{$r$-\loc}
if $r$ is matched to a point in $Y$ that is adjacent to it.
\begin{lemma}\label{lem:loc-match}
For any two sets of points $X$ and $Y$ such that $|X|=|Y|$ in a tree and a given point $r \in X$,
there exists a minimum cost perfect matching of $X$ to $Y$ that is $r$-{\loc}.
\end{lemma}
\begin{proof}
Let $M=M(X,Y)$ be any minimum cost perfect matching of $X$ to $Y$.
If $M$ is not $r$-{\loc}, do the following.  Let $y$ be the point in $Y$ that $r$ is matched to.
Then there are other points from $Y$ on the path $\cal{P}(y,r)$; let $y'$ be the one that is closest to $r$
(the last one on the path to it).  It follows that $y'$ is local for $r$.
Let $x \in X$ be the point that $y'$ is matched to.  Turn $M$ into another perfect matching $M'$
by matching $r$ to $y'$ and $x$ to $y$.  Since $y' \in \cal{P}(y,r)$,
we have $d(y',r) = d(y,r) - d(y',y)$.  Moreover, by the triangle inequality
$d(x,y) \leq d(x,y') + d(y',y)$.  Thus, the cost of $M'$ is no larger than that of $M$,
and hence $M'$ is an $r$-{\loc} minimum cost perfect matching of $X$ to $Y$.
\qed\end{proof}

We denote a (fixed) $r$-{\loc} minimum cost matching of $X$ to $Y$ by $M_r(X,Y)$.
In particular, for the line metric, we call the $r$-{\loc} matching that is
the result of the procedure described in the proof of Lemma~\ref{lem:loc-match}
applied to $\M(X,Y)$ an \emph{$r$-canonical matching}, and denote it by $\M_r(X,Y)$.

\subsection{Algorithm transformation}
Algorithm~\ref{alg:lazy-transform} describes how to transform any algorithm $\A$
into an equivalent one that is {\laz} and never incurs larger cost than $\A$.
Later, we observe that (wlog) the output algorithm is also {\loc}.
Finally, we note that for {\DC} the transformation also preserves monotonicity.

\begin{alg}[tb]
\textbf{Input:} An algorithm $\A$ and an online sequence of requests $r_1,\ldots,r_i$.\\
\textbf{Output:} A {\laz} algorithm $\A'$.

\begin{enumerate}
  \item Let $\St{\A}$ be the state of $\A$ after serving $r_i$ and $\St{\A'}$ the state of $\A'$ right before
  	serving $r_i$; note that $r_i \in \St{\A}$.
  \item Find $M=M(\St{\A},\St{\A'})$ and serve $r_i$ with the server matched to it in $M$.
\end{enumerate}
\caption{The algorithm transformation.}\label{alg:lazy-transform}
\end{alg}

In general, the transformation only guarantees that its output is a {\laz} algorithm of no larger cost than $\A$.
The additional properties rely on particular structure of the input algorithm and the matching that is used upon requests.
\begin{lemma}\label{lem:transformation}
Given an algorithm $\A$, the transformation produces a {\laz} algorithm $\A'$ such that
for every sequence of requests $R=r_1,\ldots,r_i$, it holds that $C\lp A', R\rp\leq C\lp A, R\rp$.
\end{lemma}
\begin{proof}
It follows directly from the transformation that $\A'$ is {\laz}.  In order to bound its cost,
denote the costs of $\A$ and $\A'$ upon serving request $r_i$ by $C_i(\A)$ and $C_i(\A')$ respectively. 
We will prove that
\begin{equation}\label{eq: trans-cost}
C_i(\A') + \Delta\Phi \leq C_i(\A) \enspace,
\end{equation} 
where $\Phi = d(\St{\A},\St{\A'})$ is a non-negative potential function.

To prove~\eqref{eq: trans-cost}, it suffices to consider the moves of $\A$ and $\A'$
independently, in this order.  Fix $M = M(\St{\A},\St{\A'})$.

We keep $M$ fixed as $\A$ moves its servers.
Clearly, when $\A$ moves a server $\srv$ by distance $d$, the cost of $M$ does not
increase by more than $d$.  Hence, the same holds for the perfect matching of minimum cost.
Thus $\Phi$ increases by at most $d$, and \eqref{eq: trans-cost} holds.

Once $\A$ is done with its moves, we analyze the move of $\A'$.  Note that at this point
$r_i \in \St{\A}$, i.e., $\A$ has one of its servers at $r_i$.  Let $M=M(\St{\A},\St{\A'})$
be a minimum cost matching of the algorithms' servers and $\srv \in \St{\A'}$ the server of $\A'$
that is matched to $r_i$.  Upon the move of $\srv$ to $r_i$, $\Phi$ decreases by at least $d(\srv,r_i)$
since this is the case if we keep $M$ fixed.  (It is actually easy to see that there is no cheaper matching).
Thus, $C_i(\A') + \Delta\Phi \leq 0$ holds.
\qed\end{proof}

\begin{remark}\label{rem:trans-loc}
The transformation produces a {\laz} {\loc} algorithm if, upon each request $r_i$, an $r_i$-{\loc} matching is used.
\end{remark}

Note that we have proved that the algorithm that our transformation produces is {\loc}
but not that it is {\m}.
Nevertheless, for the case of line and the \emph{Double Coverage} ({\DC})
algorithm~\cite{ChrobakKPV91,ChrobakL91}, we can prove monotonicity.

\subsection{Transforming {\DC} on a line}

The \emph{Double Coverage} (\DC) algorithm has been originally designed
for the line metric~\cite{ChrobakKPV91}, and later generalized
to all tree metrics~\cite{ChrobakL91}.  For these metrics, it attains
the optimal competitive ratio of $k$.  For completeness, we describe {\DC}
for the case of the line metric in Algorithm~\ref{alg:dc}.

\begin{alg}[tb]
\textbf{Input:} A request $r$ on the line.

\begin{enumerate}
  \item Let $\srv_L$ and $\srv_R$ be the servers of {\DC} immediately to the left of $r$
  	and right of $r$ respectively.  (Note that one of them may not exist; in such case ignore it.)
  \item Move both $\srv_L$ and $\srv_R$ towards $r$ by $\min\{d(\srv_L,r),d(\srv_R,r)\}$.
  \item Serve $r$ with the server that reached it. (Note that both $\srv_L$ and $\srv_R$ could have.)
\end{enumerate}
\caption{Double Coverage on a line.}\label{alg:dc}
\end{alg}

To prove that {\DC} on a line can be transformed to a {\laz} {\loc} {\m} algorithm,
we will rely on the fact that {\DC} itself is {\loc} and {\m}, as well as on properties
of canonical matchings on a line.

\begin{lemma}\label{lm:wfdc-is-monotone}
Applied to {\DC} on a line, the transformation produces a {\laz} {\loc} {\m} algorithm
if, upon each request $r_i$, an $r_i$--canonical matching is used.
\end{lemma}
\begin{proof}
We focus on monotonicity since the other properties follows from Lemma~\ref{lem:transformation} 
and Remark~\ref{rem:trans-loc}.
To this end, let $\DC'$ be the output of the transformation and $\St{\DC'}=\{\srv_1,
\ldots,\srv_k\}$, where $\srv_1 \leq \srv_2 \leq \ldots \leq \srv_k$, its current state.  Since for each request $r$, we use the
$r$-{\loc} canonical matching in the transformation, $r$ will be served by a server of $\DC'$ that is local for it.
In particular, $\DC'$ will serve a request $r \leq \srv_1$ with $\srv_1$ and a request $r \geq \srv_k$ with $\srv_k$.
Similarly, $\DC'$ will serve a request at position of one its servers $\srv$ with $\srv$.

Now consider $\srv_i < r < \srv_{i+1}$, i.e., a request between two adjacent servers of $\DC'$.
Suppose that $\DC$ serves $r$ with its $j$-th leftmost server.
Recall that an $r$-{\loc} canonical matching between the state of $\DC$ after serving $r$ to $\St{\DC'}$
would be used in the transformation for a request (at) $r$.
In the canonical matching $r$ would be matched to $\srv_j$ but as this matching is made {\loc},
$r$ is matched to the last server of $\DC'$ on the path from $\srv_j$ to $r$.
In other words, $r$ is matched to $\srv_i$ if $j \leq i$ and to $\srv_{i+1}$ otherwise,
and it is served by the server it is matched to.
Since $\DC$ itself is {\loc} {\m}, $j$ is a non-decreasing function of $r$, which implies that $\DC'$ is {\m}.
\qed\end{proof}

\subsection{Payment scheme}

We define a payment scheme as $P\colon \St{\A} \mapsto \Reals$.
The disutility of an agent while serving a request $r$ using server $\srv$ is defined as $d(r,\srv) + P(\srv)$.
We assume that at each iteration the request $r$ is served using the server $\srv$ that minimizes
$d(r,\srv) + P(\srv)$, where $r$ is the request the incoming agent wishes to serve.

\begin{theorem}\label{thm:pricing-k-server-trees}
For every {\laz} {\loc} {\m} deterministic $k$-server algorithm $\A$ on a tree,
there exists a payment scheme $P$ that makes incoming agents \emph{weakly} follow $\A$
and incur movement cost no larger than that of $\A$.
\end{theorem}
See the remarks after the proof for explanation what ``weakly'' means.
\begin{proof}
We obtain the desired pricing $P$ by fixing the price of an arbitrary server in $\St{\A}$
and setting the remaining prices so that they satisfy a system of linear equations.
The equations follow from the following consideration.

Let $\St{\A}$ be the state of the $\A$.
Since $\A$ is {\laz} {\loc} {\m}, for every server $\srv \in \St{\A}$,
there is subtree containing $\srv$ consisting of all the points
that would be served by $\srv$ if a request were made there.
Clearly, these subtrees form a partition of the underlying tree.
From now on we refer to these subtrees as regions that particular servers
are responsible for to avoid confusion. Since the underlying space is a tree,
the adjacency relation on the regions defines a subtree spanning the regions.

Let $\srv$ and $\srv'$ be a pair of servers whose regions are adjacent.
Then, there is a \emph{threshold point} $v(\srv,\srv')$
(which may be one of $\srv$, $\srv'$) on the path $\cal{P}(\srv,\srv')$ such that,
ignoring $v(\srv,\srv')$ itself, the algorithm would serve all requests
on $\cal{P}(\srv,v(\srv,\srv'))$ with $\srv$ and all requests on
$\cal{P}(\srv',v(\srv,\srv'))$ with $\srv'$; a request at $v(\srv,\srv')$
would be served with either of $\srv$, $\srv'$.
Having identified the threshold vertex $v(\srv,\srv')$, we introduce the following equation:
\begin{equation}\label{eq:price-system}
 P(\srv) + d(\srv,v(\srv,\srv')) = P(\srv') +d(\srv',v(\srv,\srv')) \enspace.
\end{equation}

We claim that the system of equation(s)~\eqref{eq:price-system}
taken for all pairs of adjacent servers has a unique solution
once the price for an arbitrary server is fixed.
This follows, since we introduced one equation per edge of the tree of regions;
in particular, acyclicity guarantees that the equations are independent.
In fact, the pricing that is the unique solution to our system can
be found by a BFS traversal of $G$, starting at the vertex with fixed
price, and using~\eqref{eq:price-system} at each newly visited node
to determine its price.

To see that the pricing $P$ makes the agents follow the algorithm,
consider a request at point $v$ that $\A$ would serve with $\srv$.
Then for any $\srv' \in \St{\A}$,  consider the path from $\srv'$ to $v$,
and let $\srv'=\srv_0,\srv_1,\ldots,\srv_i=S$ be the sequence of servers
responsible for successive regions on it.   By~\eqref{eq:price-system},
for each $j<i$, we have
$P(\srv_j) = P(\srv_{j+1}) + d(\srv_{j+1},v(\srv_j,\srv_{j+1})) - d(\srv_{j},v(\srv_j,\srv_{j+1}))
\geq P(\srv_{j+1}) - d(\srv_{j+1},\srv_j)$.
Summing the inequalities over a prefix of the path, we get that for all $j\leq i$
\begin{equation}\label{eq:path-ineq}
P(\srv')=P(\srv_0) \geq P(\srv_{j}) - d(\srv_0,\srv_j) \enspace.
\end{equation}
Therefore, if the server $\srv=\srv_i$ itself lies on the path, \eqref{eq:path-ineq} yields that
$P(\srv') + d(\srv',v) = P(\srv') + d(\srv',\srv) + d(\srv,v) \geq P(\srv) + d(\srv,v)$,
i.e., that $\srv$ minimizes the disutility among servers responsible for the regions
along the path.  If $\srv$ does not lie on this path, i.e., if $v$ lies between
$\srv'$ and $\srv$, we note that by~\eqref{eq:path-ineq} for $j=i-1$,
$P(\srv') + d(\srv',v) = P(\srv') + d(\srv',\srv_{i-1}) + d(\srv_{i-1},v) \geq P(\srv_{i-1}) + d(\srv_{i-1},v)$.
But as the region $\srv_{i-1}$ is responsible for is adjacent to the one $\srv$ is responsible for,
\eqref{eq:price-system} guarantees that
$P(\srv) + d(\srv,v) \leq P(\srv) + d(\srv,v(\srv,\srv_{i-1})) = P(\srv_{i-1}) +d(\srv_{i-1},v(\srv,\srv_{i-1})) \leq P(\srv_{i-1}) + d(\srv_{i-1},v)$.
Thus $\srv$ minimizes the disutility also in this case, and thus overall.
\qed\end{proof}

As noted before, we stress that tie-breaking is up to the agents,
and therefore we cannot guarantee that the agents will follow an algorithm exactly.
This is what we mean by saying in Theorem~\ref{thm:pricing-k-server-trees} that the
agents ``weakly follow'' such algorithm.
%
%
However, tie-breaking affects neither Lemma~\ref{lem:transformation}
nor Remark~\ref{rem:trans-loc}, i.e., even if agents break ties differently than the algorithm,
in particular, even if they move non-{\loc} servers, we can still price the servers
in such a way that the cost of moving servers by the agents is no larger
than the cost of moving them by the algorithm.

We sketch why this is the case.  If for a given request $r$,
more than one adjacent server minimizes the  disutility of serving $r$,
then $r$ is exactly at the threshold vertex of the pricing scheme's design,
i.e., the boundary of at least two regions.
Then it is easy to see that either server can be matched to $r$
in a minimum cost matching, i.e., a move of either server could result
from the transformation, depending on the matching that is chosen.
To deal with non-adjacent servers, it is easiest to assume that for
any pair of adjacent servers $s_i$ and $s_{i+1}$,
the threshold vertex $v(s_i,s_{i+1})$ is always strictly between them.
This can be done by moving the vertex by arbitrarily small $\epsilon>0$
away from the server $s$ it coincides with.
Clearly, this may result in violating~\eqref{eq: trans-cost}
only by $2\epsilon$, and we can make $\epsilon$ decrease exponentially
for successive requests, so that the total additional additive cost
is at most $c\cdot\epsilon$, where $c$ can also be made arbitrarily small.

\begin{corollary}
    There exists a payment scheme for the $k$-server problem on a line
    which yields a price of anarchy of $k$.
\end{corollary}

\section{Pricing Parking Slots on a City Block}
\subsection{Problem Definition}

We consider the following strategic parking problem;
the underlying problem is known as the (online) metric matching
but we use the parking terminology.

 The ``town" is modeled as an undirected graph $G=(V,E)$, $\la V\ra=m$. Vertices of $G$ represent empty parking slots, which are the possible destinations of drivers arriving in town.  Every edge $e=(u,v)\in E$ has an associated weight $w_e$, corresponding to the distance between vertices $v$ and $u$.

As for the cars:
     \begin{itemize}
     \item A total of $n$ cars are going to park over time ($n\leq m$). Cars do not depart until every car has parked.
    \item Cars arrive online, one at a time, in order of their index, $1, 2, \ldots, n$.
    \item Cars can only park at unoccupied parking spots.
    \item Cars behave selfishly and seek to minimize their disutility (defined below).
    \end{itemize}

We denote the parking spot taken by car $i$ by $x(i)$, and call the function
$x:[n]\mapsto V$ that maps cars to parking slots {\em an assignment}.
    
A dynamic pricing scheme $P = (P_1, P_2, \ldots, P_n)$, where
\[P_i:\{x(1), x(2), \ldots, x(i-1)\} \times V\mapsto \Reals \enspace,\]
sets a price for an unoccupied parking space $v\in V$ given that there are $i-1$ occupied parking spaces: $x(1), \ldots, x(i-1)$.  The functions $P_i$ may use random bits. Set $P_0(v)=0$ for all $v\in V$ (or any other constant). When we refer to $P_i(v)$,
we assume that $i-1$ cars have already arrived and are parked at $x(1), \ldots, x(i-1)$.  Thus,

    \begin{itemize}
       \item Given $G$ and $i\in 1, \ldots, n$, and given that cars $j\in 1, \ldots, i-1$ are already parked at locations $x(1), \ldots, x(i-1)$, $P_i(v)\in \Reals$ gives a price to park at $v\in V\setminus \{x(1), \ldots, x(i-1)\}$. This is a one-time parking fee, not (say) charged by the hour.
		\item   The current parking fees, $P_i(v)$, $v\in V \setminus \{ x(1), \ldots, x(i-1)\}$ are known to car $i$ {\em before} parking.
   \end{itemize}

     Incentives for cars are determined as follows:
    \begin{itemize}
    \item Let $g(i)$, $1 \leq i \leq n$ the be destination of car $i$, $g$ is called the {\em goal function}.
    \item Given the graph $G$, goal function $g$, and assignment function $x$, the {\em non-monetary cost} to agent $i\in[n]$ with destination $g(i)$ and an assigned parking slot $x(i)$ is $d\left( g\left( i\right),x\left( i\right)\right)$, the weight of a minimum weight path from $x(i)$ to $g(i)$ in $G$.
    \item The disutility of agent $i$ (who chooses to park at unoccupied slot $x(i)$) is $d\left( g\left( i\right),x\left( i\right)\right)+P_{i}(x(i)).$  {\em I.e.}, the disutility is the sum of the non-monetary cost to agent $i$ plus the parking fee to park at $x(i)$.
        \end{itemize}

	Therefore, the agent will choose to park at:
	$$s_i=\argmin_{s\in S_i}\{d\left( g\left( i\right),s\right)+P_{i}(s)\},$$
	where $S_i$ is the set of unoccupied vertices of $G$ at the time agent $i$'s arrival.

	As the goal is cost minimization, the target function $f$ is the social cost,
	$$f(\overline{\sigma},\overline{s}) = \sum_{i=1}^{|\overline{\sigma}|} d\left( g\left( i\right),s\right),$$
	{\em i.e.}, the sum of all the agents' distances from their goals.
	
	Since the pricing schemes presented in this section are randomized, the price of anarchy, defined in Equation (\ref{eq:poa}) should be taken in expectation, and is defined as
	\begin{equation}
	\max_{\overline{\sigma},\, \overline{s}\in \mathrm{eq}(\overline{P},\overline{\sigma})}\frac{\E\big(f(\overline{\sigma},\overline{s})\big)}{  f(\overline{\sigma},\overline{\mathrm{opt}(\sigma)})}, \label{eq:epoa}
	\end{equation}
	where $\mathrm{eq}(\overline{P},\overline{\sigma})$ is the set of equilibria assignments of $\overline{P}$, which depends on the random bits used in $\overline{P}$.

We refer to the pricing scheme $P$ that satisfies $P_i(v)=0$ for all $i\in 1, \ldots, n$, all $x(1), \ldots, x(i-1)$ and all $v\in V\setminus \{x(1), \ldots, x(i-1)\}$ as the \emph{free parking},
and its price of anarchy as the \emph{cost of free parking}.


\subsection{Dynamic Harmonic Pricing for the Weighted Line Graph} \label{sec:lineprice}
In this section we describe the \emph{harmonic pricing scheme} that guaranties an $O\left(\min\left\{\log R, n^2\right\}\right)$ approximation to the minimal cost matching \emph{in expectation}, where $R$ is the aspect ratio of the underlying metric. 
The expectation is over the random bits of the mechanism. The coordination mechanisms sets fees to the vacant parking slots.
The payments are devised without knowing the goals of the drivers. Once set, a car parks so as to minimize its disutility.

The improved approximation guaranties of this mechanism should be compared with the $\Omega(2^n)$ approximation lower bound for the greedy matching on the weighted line graph with free parking \cite{KalyanasundaramP93,KhullerMV94}. On the same input, the \emph{harmonic pricing scheme} incurs a cost of $n$ in expectation.

The pricing is designed to emulate the Harmonic Algorithm of Gupta and Lewi~\cite{GuptaL12}
for a line metric. This simple algorithm works as follows: let $d_l$ be the distance of the arriving car 
from the nearest free parking slot to the left of it, and $d_r$ be its distance from the nearest free 
parking slot to the right. The algorithm assigns the car to the slot on the left with probability $\frac{d_r}{d_l+d_r}$, and to the slot on the right with the remaining probability. 
Gupta and Lewi proved the following.
\begin{theorem}[\cite{GuptaL12}]\label{thm:gupta}\label{harmonic}
Let $d_{\mathrm{max}}$ be the maximal distance between two vertices, and let $d_{\mathrm{min}}$ be the minimal one. The competitive ratio of the harmonic algorithm on a line metric is $O\left( \log{\frac{d_{\mathrm{max}}}{d_{\mathrm{min}}}}\right)$.
\end{theorem}

Theorem~\ref{thm:gupta} yields an $O(\log n)$ approximation algorithm to the min cost matching
in the case where quantity $d_{\mathrm{max}}/d_{\mathrm{min}}$,
called the \emph{aspect ratio} of the metric, can be bounded by a polynomial of $n$.
In fact, Gupta and Lewi~\cite{GuptaL12} also show how this algorithm can be combined
with the doubling technique so that effectively the aspect ratio is always bounded
by a polynomial of $n$.  This relies on continuous estimation of the optimum cost
and modifying the metric distances based on that.

We provide two pricing schemes. The price of anarchy of the simpler one matches
the competitive ratio guarantee of Harmonic as stated in Theorem~\ref{thm:gupta}.
The more sophisticated pricing scheme matches the refined bound of $O(\log n)$
given an estimate of the (final) optimum social cost a priori.

\subsubsection{Computing the posted parking fees}

We now show how to set prices to the vacant parking slots that effectively makes incoming [selfish] agents park {\em as though} they were obeying orders given by the Harmonic Algorithm.
As the algorithm is randomized, so will be the pricing.

We say that a pricing scheme \emph{ensures harmonic behavior} if, when an agent arrives:
\begin{itemize}
    \item She chooses to park at her goal if vacant.
    \item Otherwise, she chooses between the two closest vacant slots to the left and right of her goal with harmonic probabilities.
\end{itemize}

A maximal contiguous run of occupied parking spaces is called a \emph{block}.

Just before car $i$ is set to park,
let $V'=\{v_1,\dots ,v_k\}$ the set of vacant parking spots and let $B_1,B_2,\ldots, B_{\ell}\subseteq V$ the current set of ``blocks".
We use $L\left( B_j\right) $ to denote the first available vertex to the left of block $B_j$, for some $j\in \{1,\ldots,\ell\}$, and let $R\left( B_j\right)$ denote the first available vertex to the right of $B_j$. Let $d_j = d\left( L\left( B_j\right), R\left( B_j\right)\right)$, we define a set of distributions $D_1,D_2,\ldots,D_\ell$, where for every $j\in \{1,\ldots,\ell\}$, $D_j$ is uniformly distributed on the real interval $\left[ -d_j,d_j\right]$.

Let $P:V'\mapsto \Reals^+$ be a random function that sets prices to vacant slots. Previously, this function was denoted $P_i$, but we omit the index $i$ for the remainder of this section.

We say that $P$ satisfies the \emph{Harmonic Payment Conditions} if:
\begin{enumerate}\label{itm:GuptaLewi}
\item For any $j\in\{1,\ldots,\ell\}$, $P\left( L\left( B_j\right)\right)- P\left( R\left( B_j\right)\right)\sim D_j$.
\item For any $u,v\in V'$ s.t. there are no unoccupied parking slots between them: $-d \left( u,v \right) < P(u) - P(v) < d\left( u,v\right)$.
\end{enumerate}

\begin{lemma}
\label{lem:har-pay-cond}
Any pricing function $P$ that satisfies the Harmonic Payment Conditions, ensures harmonic behaviour on the part of arriving cars, regardless of their goals.
\end{lemma}
\begin{proof}
    First, notice that the second condition of the harmonic payment conditions implies that for any $u,v\in V'$, $-d \left( u,v \right) < P(u) - P(v) < d\left( u,v\right)$, even if there are vacant slots between them. This implies that if an agent's goal, $g$, is vacant, then for all $v\in V' \setminus \{g\}$, we have that $P(g)<P(v)+d\left( g,v\right)$. Ergo, if the agent's goal is vacant, she parks there.

    Now consider an agent whose goal $v$ is blocked, i.e., $v\in B_j$ for some $j\in\{1,\ldots,\ell\}$. 
Let $v_L$ denote $L\left( B_j\right)$ and $v_R$ denote $R\left( B_j\right)$. For every vacant slot $v'\in V'$ left of $v_L$ we have $d\left( v, v_L\right) +P\left( v_L\right) < d\left( v, v_L\right) + d\left( v_L,v'\right)+ P\left( v'\right)=d\left( v, v'\right) + P\left( v'\right)$. Symmetrically, for every vacant slot $v'$ right of $v_R$, we have that $d\left( v, v_R\right) +P\left( v_R\right) <d\left( v, v'\right) + P\left( v'\right)$. Therefore, the agent has only to choose between $v_L$ and $v_R$.

    Let $F_j$ be the density function of distribution $D_j$. Notice that $F_j\left( x\right) = \frac{x+d_j}{2 d_j}$. We have that:
    \begin{align*}
        \Pr&\left[ \mbox{The agent takes } v_L\right] \\
         &= \Pr\left[ d\left( v, v_L\right) + P\left( v_L\right) \leq d\left( v, v_R\right) + P\left( v_R\right) \right]\\
        &= \Pr\left[ P\left( v_L\right)-P\left( v_R\right) < d\left( v, v_R\right)-d\left( v, v_L\right)\right]\\
        &= \Pr\left[ P\left( v_L\right)-P\left( v_R\right) < d_j-2d\left( v, v_L\right)\right]\\
        &= F_j\left( d_j-2d\left( v, v_L\right)\right)\\
        &= \frac{2 d_j-2 d\left( v, v_L\right)}{2 d_j}\\
        &= \frac{d\left( v, v_R\right)}{d_j} \enspace,
    \end{align*}
    where the forth equality is due to the first harmonic payment condition. Therefore, we get that the agent follows the harmonic behavior.
\qed\end{proof}

Next, we show a specific pricing scheme that satisfies the harmonic payment conditions.

\begin{observation}\label{obs:sat}
Let $R(v)$ be the set $\{j\in\{1,\ldots,\ell\}: v\text{ is left of } B_j\}$ and let $q_1\sim D_1,g_2\sim D_2,\ldots, q_\ell\sim D_\ell$.
For any constant $c$, setting $P(v)= c+\sum_{j\in R(v)}{q_j}$ satisfies the \emph{Harmonic Payment Conditions}.
\end{observation}

We can use the constant $c$ in order to obtain a pricing scheme that satisfies the harmonic payment conditions in which there are no negative transfers. One can ask for a pricing scheme that optimizes an additional objective: minimizes the sum of prices of all vacant parking slots, the maximal price of a slot, the difference in pricing of the slots over time, etc. We can use an LP to satisfy such objectives.  For example, Linear Program~\ref{lp:min} minimizes the sum of prices of slots given $q_j \sim D_j$ as prescribed in Observation~\ref{obs:sat}.  Note that Linear Program~\ref{lp:min} is feasible since the payment scheme of Observation~\ref{obs:sat} satisfies all its constraints.

\begin{linprog}[tb]
$\min{\sum_{v \in V'}p_v}$ s.t.
\begin{itemize}
\item $\forall j\in \{1,\ldots, \ell\}$, $p_{L(B_j)} = p_{R(B_j)} + q_j$.
\item $\forall u,v\in V'$ s.t. there is no vacant slot between them, $-d \left( u,v\right) < p_u - p_v < d \left( u,v\right)$.
\item $\forall v\in V'$, $p_v\geq 0$.
\end{itemize}
\caption{A linear program for pricing that satisfies the harmonic payment conditions and also minimizes the sum of parking fees over all parked cars. Note that this program actually minimizes the sum of the posted prices of vacant parking spaces, but --- any feasible solution to this LP results in the same harmonic behavior so, by repeatedly computing such pricing, we minimize the sum of payments of all parked cars.}
\label{lp:min}
\end{linprog}

Since we have shown there exists a pricing scheme which makes incoming agents simulate the behavior of the harmonic algorithm, we can use this pricing scheme to coordinate incoming agents, and improve the PoA significantly. 

In graphs where the aspect ratio is bounded by some polynomial in $n$, we get:
\begin{corollary}
    \label{thm:pricingline}
    Any pricing scheme which satisfies the harmonic payment conditions gives an $O\left( \log{n}\right)$ approximation to the optimal assignment when the aspect ratio is bounded by $\mathrm{poly}(n)$.
\end{corollary}
For general weighted line graphs, without further assumptions, we get the following:
\begin{theorem}
    \label{thm:pricingweightedline}
    Any pricing scheme which satisfies the harmonic payment conditions gives an $O\left(\min\{ \log{\frac{d_{\mathrm{max}}}{d_{\mathrm{max}}}}, n^2\}\right)$ approximation to the optimal assignment in weighted line graphs in expectation.
\end{theorem}
\begin{proof}
The proof that the harmonic algorithm is $O\left( \log{\frac{d_{\mathrm{max}}}{d_{\mathrm{max}}}}\right)$ competitive is given in~\cite{GuptaL12}.
For proving that the competitive ratio is also bounded by $O(n^2)$, we use the same notation as in \cite{GuptaL12}%
\footnote{\sloppy A full version is available at
\url{http://theory.stanford.edu/~klewi/papers/metric-matching-full.pdf}
but its analysis of Harmonic is slightly different.}.

Replacing the final bound on $\E[\delta_{max} | \delta_1]$ in the proof of Lemma~4~\cite{GuptaL12}
by the following one gives us the desired result.
\begin{lemma}
\label{lm:newhybridlemma}
$\E[\delta_{max} | \delta_1] \leq n^2 \delta_1$.
\end{lemma}

\begin{proof}
First note that $\delta_{max}$ is the distance between two servers on the line, therefore there are at most $n^2$ possible values for this distance.
Let $J$ be the set of possible distances.
We get that
\begin{align*}
 \E_\pi [ \delta_{max} | \delta_1 ] &= \sum_{k \in J} k \Pr[\delta_{max} = k | \delta_1] \leq \sum_{k \in J} k \Pr[\delta_{max} \geq k | \delta_1] \\
 &= \sum_{k \in J} k Q_{\delta_1, k} \leq \sum_{k \in J} k \frac{\delta_1}{k}  = \sum_{k \in J} \delta_1 = |J|\delta_1 \\
 &\leq n^2 \delta_1 \enspace,
\end{align*}
where $Q_{\delta_1, k} \leq \frac{\delta_1}{k}$ by Lemma~4(ii) of~\cite{GuptaL12}.
\qedd\end{proof}\end{proof}
\subsubsection{Pricing with a Prior}
In this section we show how to improve the expected PoA via a dynamic pricing scheme when the city planner has an estimate of the cost of the optimal matching. This seems to be a reasonable assumption in many scenarios, since the city planner only needs to approximate the optimal matching to within a $\mathrm{poly}(n)$ factor in order to devise a pricing scheme that achieves an $O(\log n)$ approximation. This is an improvement over the \emph{harmonic pricing scheme} whenever the aspect ratio is superpolynomial in $n$. 

Gupta and Lewi show that given $Z\in \left[ \mathrm{OPT}, c\cdot \mathrm{OPT}\right]$ where $c>1$, one can get an $O\left(\log{cn}\right)$ approximation using the following metric change:
\begin{itemize}
    \item Disconnect every edge with weight $\geq Z$.
    \item Change every edge of weight smaller than $Z/(2c\cdot n^2)$ to be of weight \emph{exactly} $Z/(2c\cdot n^2)$.
\end{itemize}
Let $d'$ denote the new metric after the transformation. First, note that the aspect ratio in $d'$ is $2c\cdot n^3$, since the maximal distance is at most $Zn$, while the minimal one is at least $Z/(2c\cdot n^2)$. Another key observation is that since the optimal matching doesn't use any disconnected edges and the weight of each edge increased by at most $Z/(2c\cdot n^2)$, the cost of every matched server-request in the optimal matching increased by at most $Z/(2c\cdot n)$. Therefore, if OPT denotes both the optimal
matching in $(V,d)$ and its cost in $(V,d)$, then the cost of OPT in $(V,d')$ is at most $\mathrm{OPT}+Z/2c\leq 1.5\mathrm{OPT}$. As the weights in $d$ are no larger than in $d'$, running the harmonic algorithm in $(V,d')$ yields an $O\left(\log{cn}\right)$ approximation to $\mathrm{OPT}$ in $(V,d)$.

Thus, given $Z$, an estimate of the cost of the optimal matching, a dynamic pricing scheme needs to simulate the change in the metric. Specifically, the pricing scheme needs to:
\begin{enumerate}
    \item Make sure agents don't ``jump over" edges of weight bigger than $Z$.
    \item Make incoming agents simulate the behavior of the harmonic algorithm according to the transformed metric by acting selfishly.
\end{enumerate}

We say that an online matching algorithm is \emph{monotone} if the following conditions hold:
\begin{enumerate}
    \item If an incoming request is at a position of an unmatched server, the server will be matched to that request.
    \item Let $B_j$ be some block of matched servers, let $v_L=L(B_j)$ and $v_R=R(B_j)$.
    Consider any two points $a,b\in B_j$ and a request to be made at either of them.
    Then
       \begin{eqnarray*}
        &  d\left( a,v_L\right)\leq d\left( b,v_L\right) \iff \\
        &  \Pr\left[\mbox{request at }a\mbox{ would be matched to }v_L\right] \\
        &  \geq \Pr\left[\mbox{request at }b\mbox{ would be matched to }v_L\right].
        \end{eqnarray*}
\end{enumerate}
It is easy to see that running the harmonic algorithm on both the original and the transformed metrics yields a monotone algorithm. The following theorem shows that every \emph{monotone} algorithm can be priced:
\begin{theorem}
    \label{thm:monpay}
    Every monotone online algorithm for the min cost matching problem on a line can be transformed into a dynamic pricing scheme which makes incoming agents simulate the algorithm.
\end{theorem}
\begin{proof}
Let $A$ be a monotone algorithm for the online min cost matching problem. Let $B_1,B_2,\ldots, B_{\ell}\subseteq V$ be the blocks that represent the current state of the graph, and let $V'=V\setminus\bigcup_{j=1}^{\ell} B_j$ the set of vacant parking slots. Let $v_L=L\left( B_j\right)$ and $v_R=R\left( B_j\right)$. Let $p_L^j:B_j\mapsto\left( 0,1\right]$ be a probability function that maps every vertex in the block into the probability of the vertex being matched by $v_L$ in $A$. Let $V_j=\{v_j^1,v_j^2,\ldots, v_j^t\}\subseteq B_j$ the set of vertices for which $p_L^j\left( v_j^1\right) > p_L^j\left( v_j^2\right) > \ldots> p_L^j\left( v_j^t\right)$. We define a distribution $D_j$ which determines the difference between $P\left( v_R\right)$ and $P\left( v_L\right)$. For every $i\in \{1,\ldots, t\}$, let $\Delta_i=d\left( v_j^i, v_L\right)-d\left( v_j^i, v_R\right)$.

$D_j$ is defined via the following cumulative density function $F_j:$
\begin{eqnarray*}
F_j(\Delta) =
\begin{cases}
0 &\Delta \leq -d_j\\
1-p_L^j\left( v_j^i\right) &\Delta_i \geq \Delta > \Delta_{i-1}\\
1 & \Delta > d_j,
\end{cases}
\end{eqnarray*}
where $\Delta_0=-d_j$ and $\Delta_{t+1}=d_j$. From $A$'s monotonicity, it is clear that $-d_j\leq \Delta_1\leq \Delta_2\leq\ldots\leq \Delta_t\leq d_j$, and therefore $F_j$ is a valid density function.

We claim that a dynamic pricing scheme $P:V'\mapsto \Reals^+$ that satisfies the following two conditions makes incoming agents simulate algorithm $A$ by acting selfishly:
\begin{enumerate}
    \item For any $j\in\{1,\ldots,\ell\}$, $P\left( R\left( B_j\right)\right)-P\left( L\left( B_j\right)\right)\sim D_j$.
    \item For any $u,v\in V'$ s.t. there are no unoccupied slots between them: $-d(u,v)<P(u)-P(v)<d(u,v)$.
\end{enumerate}

Just as in Lemma \ref{lem:har-pay-cond}, the second condition yields that if an agent's goal vertex is vacant, she will occupy it. Let $v\in B_j$ be the incoming agent's goal for some $j\in \{1,\ldots,\ell\}$, and let $v_j^i$ be the first vertex in $V_j$ (the one with minimal index $i$) such that $v$ is \emph{not right} of it. We have that $\Pr\left[ A\mbox{ matches }v\mbox{ to }v_L\right]=p_L^j\left( v_j^i\right)$. The probability that the selfish agent will occupy $v_L$ is: 
\begin{eqnarray*}
    \Pr\left[v\mbox{ occupies } v_L\right] & = & \Pr\left[ d_L + P\left( v_L\right) < d_R + P\left( v_R\right) \right]\\
    & = & \Pr\left[ P\left( v_R\right) - P\left( v_L\right) > d_L - d_R \right]\\
    & = & \Pr\left[ P\left( v_R\right) - P\left( v_L\right) > \Delta_i \right]\\
    & = & 1- \Pr\left[ P\left( v_R\right) - P\left( v_L\right) \leq \Delta_i \right]\\
    & = & 1-F\left( \Delta_i\right)\\
    & = & p_L^j\left( v_j^i\right),
\end{eqnarray*}
where the third equality stems from the fact that $\Delta_i \geq d_L - d_R > \Delta_{i-1}$. 

The dynamic pricing scheme given in Algorithm~\ref{alg:monpay} satisfies the conditions for simulating algorithm $A$.
\begin{alg}[tb]
\begin{itemize}
\item Let $q_1\sim D_1,g_2\sim D_2,\ldots, q_\ell\sim D_\ell$.
\item Let $L(v)$ be the set $\{j\in\{1,\ldots,\ell\}: v\text{ is right of } B_j\}$.
\item $\forall v\in V'$, set $P(v)= c+\sum_{j\in L(v)}{q_j}$ for some constant $c$.
\end{itemize}
\caption{A dynamic pricing scheme that makes selfish agents simulate algorithm $A$.}
\label{alg:monpay}
\end{alg}
\qed\end{proof}

We conclude:
\begin{corollary}
    Given an estimate of the cost of the optimal matching by a $\mathrm{poly}(n)$ factor, there exists a payments scheme that gives an $O\left( \log n\right)$ approximation to the optimal assignment in weighted line graphs in expectation.
\end{corollary}
\begin{proof}
    Let $Z$ be an estimate, and $p$ be some polynomial such that $\mathrm{OPT}/p(n)\leq Z\leq p(n)\cdot\mathrm{OPT}$. First, apply the metric change using $Z'=Z\cdot p(n)$ as input. Afterwards, use a dynamic pricing scheme that simulates the harmonic algorithm according to the transformed metric. There exists such a pricing scheme according to Theorem \ref{thm:monpay}. Since $Z'\in \left[ \mathrm{OPT}, p(n)^2\mathrm{OPT}\right]$, this yields an $O\left( \log{\left( p(n)^2\cdot n\right)}\right)=O\left( \log n\right)$ approximation to the optimal assignment in expectation.
\qed\end{proof}

\section{Discussion} \label{sec:discuss}
\subsubsection*{Observable History}

For some arbitrary function $h$, we define the $h$-observable history after event $i-1$ to be $h(\overline{\sigma},\overline{s})$ where $\overline{\sigma}=\sigma_1, \ldots,\sigma_{i-1}$ is the sequence of previous agent types and  $\overline{s}=s_1, \ldots, s_{i-1}$ the sequence of decisions made thus far.

 If for some non-invertible function $q$ we have that
$$g(\overline{\sigma},\overline{s}) = q(g'(\overline{\sigma},\overline{s})), \quad \forall \overline{\sigma}, \overline{s},$$ we say that $g$ {\em hides more} than $g'$.

This suggests a tradeoff between the price of anarchy of a dynamic posted pricing scheme and the information available to it.


We remark that the study of $h$-observable histories also makes sense for arbitrary online problems, in the context of mechanism design (not restricted to posted prices) and in the non-strategic setting of online algorithms.

\subsubsection*{Pricing}

With posted prices, prices are determined before the event occurs.

For non-posted price mechanisms --- one issue is ``when will the price be determined?" Is the payment determined promptly (when the agent arrives) or only later? In \cite{DBLP:conf/sigecom/FriedmanP03} this issue came up in the context of an online Groves mechanism.

\subsubsection*{Commitment to pricing, Rate of Price Change}

Obvious disadvantages of dynamic pricing include (a) prices are unknown, you plan to go downtown but don't know what it will cost you, (b) you may start circling the block waiting for the prices to change. Thus it would be highly advantageous to be able to (a) commit to prices in advance and (b) slow the rate of change.

One interesting model is to have the price function $P_i$ depend, not on the events $\sigma_1, \ldots, \sigma_{i-1}$ and decisions $s_1, \ldots, s_{i-1}$, but on a shorter prefix of these sequences. This has the effect of committing to prices in advance. A natural question is how does this impact the price of anarchy?

\subsubsection*{Simulating Simple Online Algorithms}

It is not clear if or how one can convert an arbitrary online algorithm to a dynamic pricing scheme. The simulation technique that we've repeatedly used above seems to work when the algorithm is particularly simple. For example, it is not clear how to give dynamic posted prices with good price of anarchy for the problem of metrical matching on metric spaces other than the line. One desirable goal would be to find alternative online algorithms, possibly with higher competitive ratios, and try to find pricing schemes that simulate these algorithms.

\subsubsection*{Structural Relationships}

Given an online setting, it is clear that the following weakly increase:
\begin{itemize}
  \item The competitive ratio of an online algorithm.
  \item The price of anarchy achievable by a dominant strategy incentive compatible mechanism, where the true agent types are revealed.
  \item The price of anarchy achievable by a dominant strategy incentive compatible mechanism, when restricted to the current agent type and some $h$-observable history of previous agents.
  \item The price of anarchy achievable by a dynamic pricing scheme restricted to $h$-observable history.
\end{itemize}
However, it is not clear if there is any stronger relation that holds in general.
Moreover, these can be further refined, leading to additional questions, e.g.:
\begin{itemize}
  \item How does promptness of payment determination impact the price of anarchy of mechanisms?
  \item Given that a mechanism or pricing schemes uses $h$-observable history,
  how does the price of anarchy depend on properties of $h$?
\end{itemize}

\balance
\bibliographystyle{plain}
\bibliography{ref}

\newcommand\DrawTriangleMain {
	\tikzstyle{blue node}=[draw, circle, minimum size=1ex, inner sep=0, fill=blue] 
	\begin{tikzpicture}[xscale=0.44,yscale=0.8]
	\node[blue node, label=1] (upper) at (1.25, 1.3) {};            
	\node[blue node,label={[xshift=-2mm, yshift=-2mm]2}] (left) at (0,0) {};
	\node[blue node, label={[xshift=2mm, yshift=-3mm]3}] (right) at (6,0) {};
	\draw (left) -- (right) node[draw=none, fill=none,midway,below, red] {4};
	\draw (upper) -- (right) node[draw=none, fill=none,midway,above, red] {3};
	\draw (left) -- (upper) node[draw=none, fill=none,above, red, label={[color=red, xshift=-5mm, yshift=-10mm]2}] {};
	\begin{scope}[shift={(0,-2)}]
	\foreach \i/\j/\k in {0/2/2,2/4/2,4/7/3,7/10/3,10/12/2,12/14/2,14/17/3} { 
		\node[blue node] (left) at (\i, 0) {};
		\node[blue node] (right) at (\j, 0) {};
		\draw (left) -- (right) node[draw=none, below, midway, red] {\k};
	}
	
	\foreach \i/\j in {0/1,2/2,4/1,7/3,10/1,12/2,14/1,17/3} { 
		\node[blue node,label=\j] (label) at (\i, 0) {};
	}
	\end{scope}
	\begin{scope}[shift={(0,-1)}]
	\foreach \i/\j in {0.3/3.8,4/4.8,5/14.8} {
		\node[] (left) at (\i, 0) {};
		\node[] (right) at (\j, 0) {};
		\draw[-triangle 60] (left.west) -- (right.east) ;	
	}
	\end{scope}
	\end{tikzpicture}
}

\appendix
\section{Appendix}
\subsection{An Illustration of the Fractional Traversal Algorithm}\label{append:trav-ex}

The following is an example illustrating the execution of Algorithm \ref{alg:t} given three tasks as input:
\begin{example}[H]
	\begin{itemize}
		\item Task $w^1=(3,6,3)$:
		\begin{itemize}
			\item $\lambda^1_1=2/3$, {\em i.e.}, process two thirds of task $w^1$ in state $\tau_1=1$, for a work expenditure of $3\cdot 2/3=2$. Increase $j$ by one.
			\item $\lambda^1_2=1/3$, {\em i.e.}, process the remaining 1/3 of task $w^1$ in state $\tau_2=2$, for a work expenditure of $6/3=2$. At this point the algorithm will increase $j$ by one. Move to state $\tau_3=1$.
			\end{itemize}
			\item Task $w^2=(1,3,4)$:
			$\lambda^2_3=1$, for a work expenditure of $1$ in state $\tau_3=1$, at this point of time $\rho_j=\rho_3=1$.
			\item Task $w^3=(10,10,10)$:
			\begin{itemize} \item
				$\lambda^3_3=2/10$, for a work expenditure of $2$ in state $\tau_3=2$. Increase $j$ by one. Move to state $\tau_4=3$.
				\item $\lambda^3_4=3/10$, for a work expenditure of $3$ in state $\tau_4=3$. Increase $j$ by one. Move to state $\tau_5 =1$.
				\item $\lambda^3_5=2/10$, for a work expenditure of $2$ in state $\tau_5=1$. Increase $j$ by one. Move to state $\tau_6=2$.
				\item $\lambda^3_6=2/10$, for a work expenditure of $2$ in state $\tau_6=2$. Increase $j$ by one. Move to state $\tau_7=1$.
				\item $\lambda^3_7=1/10$, for a work expenditure of $1$ in state $\tau_7=1$, at this point of time $\rho_7=1$.      
				\end{itemize}\end{itemize}    
				
				
				\DrawTriangleMain
				

				\caption{The fractional traversal algorithm applied to a metrical task system with states $\{1,2,3\}$, distances $d_{1,2}=2$, $d_{1,3} = 3$, $d_{2,3}=4$, and traversal sequence $\tau= (1,2,1,3)^*$.
					The example and the underlying metric are also illustrated.}
					\end{example}

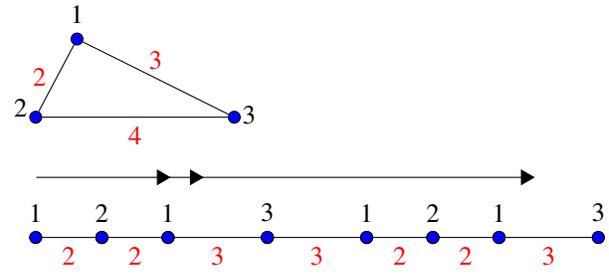
\end{document}